\documentclass[twoside,11pt]{article}
\usepackage{latexsym}
\usepackage{epsfig}
\usepackage{amsfonts}
\usepackage{amsthm,amsmath,amsfonts,amssymb}
\numberwithin{equation}{section}

\newcommand{\supp}{\operatorname{supp}}

\newcommand{\dist}{\operatorname{dist}}
\newcommand{\diam}{\operatorname{diam}}

\newcommand{\opm}{\operatorname{op}}
\newcommand{\Op}{\operatorname{Op}}

\newcommand{\edge}{\operatorname{edge}}
\newcommand{\interior}{\operatorname{int}}

\newcommand{\boldeta}{\boldsymbol \eta}
\newcommand{\loc}{\operatorname{loc}}
\newcommand{\cl}{\operatorname{cl}}
\newcommand{\rpa}{\operatorname{RPA}}
\newcommand{\fluc}{\operatorname{fluc}}
\newcommand{\comp}{\operatorname{comp}}
\newcommand{\even}{\operatorname{even}}
\newcommand{\odd}{\operatorname{odd}}
\newcommand{\om}{\operatorname{MP1}}
\def\dbar{{\mathchar'26\mkern-12mud}}

\begin{document}
\newtheorem{assumption}{Assumption}
\newtheorem{proposition}{Proposition}
\newtheorem{definition}{Definition}
\newtheorem{lemma}{Lemma}
\newtheorem{theorem}{Theorem}
\newtheorem{observation}{Observation}
\newtheorem{remark}{Remark}
\newtheorem{corollary}{Corollary}

\title{Singular analysis of RPA diagrams in coupled cluster theory}

\author{Heinz-J{\"u}rgen Flad and Gohar Flad-Harutyunyan\\
\ \\
{\small Zentrum Mathematik, Technische Universit\"at M\"unchen,
Boltzmannstr 3, D-85748 Garching, Germany}
}
\maketitle

\begin{abstract}
\noindent
Coupled cluster theory provides hierarchical many-particle models and is presently considered as the ultimate
benchmark in quantum chemistry. Despite is practical significance, a rigorous mathematical analysis of its properties
is still in its infancy. The present work focuses on nonlinear models within the random phase approximation (RPA).
Solutions of these models are commonly represented by series of a particular class of Goldstone diagrams so-called RPA diagrams.
We present a detailed asymptotic analysis of these RPA diagrams using techniques from singular analysis and
discuss their computational complexity within adaptive approximation schemes. In particular, we provide a connection
between RPA diagrams and classical pseudo-differential operators which enables an efficient treatment of the linear and nonlinear 
interactions in these models. Finally, we discuss a best $N$-term approximation scheme for RPA-diagrams and provide
the corresponding convergence rates.
\end{abstract}

\section{Introduction}
Coupled cluster (CC) theory, cf.~\cite{BM07,B91,K91,KLZ78,PL99} and references therein, provides an hierarchical but intrinsically nonlinear approach to many-particle systems which enables systematic  
truncation schemes reflecting both, the physics of the problem under consideration as well as the computational complexity
of the resultant equations. Unfortunately rigorous insights concerning the corresponding errors of approximation
are hard to achieve and a detailed understanding of the properties of solutions is still missing, see however recent progress on the subject by Schneider,Rohwedder, Kwaal and Laestadius \cite{R13,RS13,S09,LK17}. 
Within the present work, we want to tackle the asymtotic behaviour of solutions near coaslescence points of electrons. 
The corresponding asymptotic analysis for eigenfunctions of the many-particle Schr\"odinger equation has a long tradition
\cite{FHO2S05,FHO2S09,HOS81,HO292,HO2S94,Kato57}, starting with the seminal work of Kato \cite{Kato57} and culminated in the analyticity result of M.~and T.~Hoffmann-Ostenhof, Fournais and S{\o}rensen \cite{FHO2S09}.
Although CC theory orginates from the many-particle Schr\"odinger equation it is not simply possible to transfer the results to solutions of specific CC models. 
because these models represent higly nonlinear approximations
to Schr\"odinger's equation. The peculiar structure of linear and nonlinear interactions in CC models
requires a new approach based on pseudo-differential operator algebras which is introduced in the present work. We consider these operator algebras in the framework of a commonly employed series expansion of solutions in terms of so-called Goldstone diagrams, cf.~\cite{LM86} for further details. 
Our main result for a popular CC model, presented in Section \ref{RPAdiag}, affords the classification of these Goldstone diagrams in terms of symbol classes of classical pseudo-differential operators and provides their asymptotic expansions near coalescence points of electrons.  

\subsection{Nonlinear models of electron correlation}
In the following we want to focus on the SUB-2 approximation \cite{B91} and discuss the properties
of simplified models which contain the leading contributions to electron correlation. 
The quantities of interest are so called pair-amplitudes
\[
 \tau_{ij} \ : \ \left( \mathbb{R}^3 \otimes s \right) \otimes \left( \mathbb{R}^3 \otimes s \right) \ \rightarrow \ \mathbb{R} , 
 \quad (\underline{\bf x}_1, \underline{\bf x}_2) \ \rightarrow \ \tau_{ij}(\underline{\bf x}_1, \underline{\bf x}_2) ,
\]
where indices $i,j,k,l$ refer
to the occupied orbitals $1,2,\ldots N$ with respect to an underlying mean-field theory like Hartree-Fock, 
and $s := \{\alpha,\beta\}$ denotes spin degrees of freedom.
The most general model we want to consider
in the present work is given by a system of nonlinear equations for pair-amplitudes, cf.~\cite{L77}, 
which is of the form
\begin{eqnarray}
\label{MP1}
 \lefteqn{\mathfrak{Q} \left( \mathfrak{f}_1 + \mathfrak{f}_2 - \epsilon_i -\epsilon_j \right)
 \tau_{ij}(\underline{\bf x}_1, \underline{\bf x}_2) 
 = - \mathfrak{Q} V_{[i,j]}^{(2)}(\underline{\bf x}_1, \underline{\bf x}_2)} \hspace{3cm} \\ \label{ladder}
 & & - \mathfrak{Q} V^{(2)}({\bf x}_1, {\bf x}_2) \tau_{ij}(\underline{\bf x}_1, \underline{\bf x}_2)
 - \tfrac{1}{4} \mathfrak{Q} \sum_{k,l} \tau_{kl}(\underline{\bf x}_1, \underline{\bf x}_2) \langle V_{[k,l]}^{(2)}, \Psi^{(1)}_{ij} \rangle \\ \label{RPA1}
 & & + P(12/ij) \mathfrak{Q} \sum_{k} \tau_{i,k}(\underline{\bf x}_1, \underline{\bf x}_2) V_{kj}^{(1)}({\bf x}_2) \\ \label{RPA2}
 & & - P(12/ij) \mathfrak{Q} \sum_{k} \int \tau_{i,k}(\underline{\bf x}_1, \underline{\bf x}_3) V_{kj}^{(2)}(\underline{\bf x}_3, \underline{\bf x}_2) \, d\underline{\bf x}_3 \\ \label{RPA3}
 & & - \tfrac{1}{2} P(12/ij) \mathfrak{Q} \sum_{k,l} \iint \tau_{i,k}(\underline{\bf x}_1, \underline{\bf x}_3) V_{[k,l]}^{(2)}(\underline{\bf x}_3, \underline{\bf x}_4)
 \tau_{l,j}(\underline{\bf x}_4, \underline{\bf x}_2) \, d\underline{\bf x}_3 d\underline{\bf x}_4 .
\end{eqnarray}
In the following, we will frequently refer to the equation numbers (\ref{MP1}) to (\ref{RPA3}), however, depending on the context with different meanings.
Either we refer to the individual term on the right hand side or to the whole equation which includes on the right hand side all terms up to the
specified number. 
In order to simplify our discussion, let us assume in the following smooth electron-nuclear potentials, originating
e.g.~from finite nucleus models. 

Before we enter into a brief discussion of the underlying physics
of the equations let us start with some technical issues.
In order to keep formulas short we use in (\ref{RPA1}),(\ref{RPA2})and (\ref{RPA3}) the permutation operator
\[
 P(12/ij) := 1 + (21)(ji) - (12)(ji) - (21)(ij) .
\]
Equations for pair-amplitudes are formulated on a subspace
of the underlying two-particle Hilbert space which is characterized
by the projection operator 
\begin{equation}
 \mathfrak{Q} :=(1-\mathfrak{q}_1)(1-\mathfrak{q}_2) \ \   
 \mbox{with} \ \mathfrak{q} := \sum_{i=1}^N |\phi_i \rangle \langle \phi_i | ,
\end{equation}
where $\phi_i$, with $i=1,2,\ldots,N$, represent occupied orbitals.
The operator $ \mathfrak{Q}$ is commonly known as strong orthogonality operator \cite{SJMZ}.
Due to the presence of this operator, pair-amplitudes rely on the constraint
\begin{equation}
 \mathfrak{Q} \tau_{ij}(\underline{\bf x}_1, \underline{\bf x}_2) =
 \tau_{ij}(\underline{\bf x}_1, \underline{\bf x}_2) .
\label{qtau}
\end{equation}
The physical reason behind $\mathfrak{Q}$ is Pauli's principle which excludes the subspace assigned
to the remaining $N-2$ particles from the Hilbert space of the pair and an orthogonality constraint between the mean field part 
\[
 \Psi^{(1)}_{ij}(\underline{\bf x}_1, \underline{\bf x}_2) := \phi_i(\underline{\bf x}_1) \phi_j(\underline{\bf x}_2) - \phi_j(\underline{\bf x}_1) \phi_i(\underline{\bf x}_2) 
\]
and the corresponding pair-amplitude $\tau_{ij}$.

In order to set the equations for pair-amplitudes into context
let us first consider the case of Eq.~(\ref{MP1}), where the parts (\ref{ladder}), (\ref{RPA1}), (\ref{RPA2}) and (\ref{RPA3}) have been neglected.
What remains in this case is first-order M{\o}ller-Plesset
perturbation theory which provides second- and third-order corrections to the energy. Going further to Eq.~(\ref{ladder}) one recovers
the dominant contributions to short-range correlation, so-call particle ladder diagrams. Neglecting interactions between different electron pairs  
altogether in (\ref{ladder}) to (\ref{RPA3}), one recovers 
the Bethe-Goldstone equation 
\[
 \mathfrak{Q} \left( \mathfrak{f}_1 + \mathfrak{f}_2 - \epsilon_i -\epsilon_j \right)
 \tau_{ij}(\underline{\bf x}_1, \underline{\bf x}_2) 
 = - \mathfrak{Q} V_{[i,j]}^{(2)}(\underline{\bf x}_1, \underline{\bf x}_2)
 - \mathfrak{Q} V_{\fluc}^{(ij)}({\bf x}_1, {\bf x}_2) \tau_{ij}(\underline{\bf x}_1, \underline{\bf x}_2) ,
\]
here the fluctuation potential is given by
\begin{equation}
 V_{\fluc}^{(ij)}(\underline{\bf x}_1, \underline{\bf x}_2) := \frac{1}{|{\bf x}_1 - {\bf x}_2|}
 - \mathfrak{v}^{(i,j)}_{Hx}(\underline{\bf x}_1) - \mathfrak{v}^{(i,j)}_{Hx}(\underline{\bf x}_2) 
 + \tfrac{1}{2} \langle V_{[i,j]}^{(2)}, \Psi^{(1)}_{ij} \rangle
\label{Vfluc}
\end{equation}
where
\begin{equation}
 \mathfrak{v}_{Hx}^{(i,j)} := \mathfrak{v}_H^{(i)} + \mathfrak{v}_H^{(j)} + \mathfrak{v}_x^{(i)} + \mathfrak{v}_x^{(j)} 
\end{equation}
represents the contribution of orbitals $i,j$ to the Hartree and exchange potential, respectively.
The asymptotic behaviour of solutions for both models can be directly derived from the structure of asymptotic parametrices of Hamiltonian operators and has been already discussed in Ref.~\cite{FHS15} in some detail. 

The three remaining terms (\ref{RPA1}), (\ref{RPA2}) and (\ref{RPA3}) of the effective pair-equation take part in the so called
{\em random phase approximation} (RPA) which is essential for the correct description of long-range correlations.
In the following we want to consider also some  
simplified versions of RPA where exchange contributions are neglected. Such models are of particular significance with
respect to applications of RPA as a post DFT model, cf.~\cite{F08}. It has been 
shown that these models are equivalent to solving particular
variants of CC-RPA equations, cf.~\cite{SHS08}.
The terms   (\ref{MP1}), (\ref{ladder}),  (\ref{RPA1}), (\ref{RPA2}) and (\ref{RPA3}) contain various effective interaction potentials
\begin{equation}
 V_{kj}^{(1)}({\bf x}_2) := \int \frac{1}{|{\bf x}_2 - {\bf x}_3|} \phi_k(\underline{\bf x}_3) \phi_j(\underline{\bf x}_3) \, d\underline{\bf x}_3 ,
\label{Vkl1}
\end{equation}
\begin{equation}
V^{(2)}({\bf x}_1, {\bf x}_2) := \frac{1}{|{\bf x}_1 - {\bf x}_2|}  ,
\label{V2}
\end{equation}
\begin{equation}
V_{ij}^{(2)}(\underline{\bf x}_1, \underline{\bf x}_2) := \phi_i(\underline{\bf x}_1) V^{(2)}({\bf x}_1 - {\bf x}_2) \phi_j(\underline{\bf x}_2) ,
\label{Vkl2}
\end{equation}
\begin{equation}
 V_{[i,j]}^{(2)}(\underline{\bf x}_1, \underline{\bf x}_2) := V_{ij}^{(2)}(\underline{\bf x}_1, \underline{\bf x}_2) - V_{ji}^{(2)}(\underline{\bf x}_1, \underline{\bf x}_2) .
 \label{Vk12b} 
\end{equation}
The RPA model considered in the present work is still incomplete
in the sense that various terms present in the full SUB-2 model
are still missing. These missing nonlinear terms, however, do not
contribute anything new from the point of view of the following asymptotic singular analysis.  
Taking them into account would only render our presentation unnecessarily complicated.
Therefore our RPA model represents a good choice in order 
to study the properties of pair-amplitudes in some detail and
in particular the effect of non linearity which enters into our
model via the coupling term (\ref{RPA3}). The nonlinear character of our model is not quite of the form familiar from the theory
of nonlinear partial differential equations, cf.~\cite{GT}, and we will therefore consider
an unconventional approach to tackle the problem. Our approach reflects not only the particular character of the various coupling terms but also the singular structure of the interactions and
pair-amplitudes. The RPA terms (\ref{RPA2}) and (\ref{RPA3})
resemble to compositions of kernels of integral operators
and it is tempting to consider these terms in the wider context
of an appropriate operator algebra. It will be shown in the following, that the algebra of classical pseudo-differential operators provides a convenient setting. As complementary
approach let us study pair-amplitudes in the framework of weighted Sobolev spaces with asymptotics which we have already
considered in the context of singular analysis in order to
determine the asymptotic behaviour near coalescence points of electrons, cf.~Ref.~\cite{FHS15}. Both seemingly disparate
approaches complement one another in the asymptotic singular analysis of RPA models.

\subsection{Iteration schemes and their diagrammatic counter parts}
\label{iterationscheme}
The Bethe-Goldstone equation and various nonlinear RPA models, discussed above, are commonly solved in an iterative manner. 
Before we delve into the technicalities of our approach let us briefly discuss iteration schemes and their physical interpretation
in a rather simple and informal manner in order to outline certain essential features of the present work.
To simplify our notation, occupied orbital indices $i,j,k,l$ and spin degrees of freedom which appear in pair-amplitudes and interaction potentials have been dropped because they are not relevant in the following discussion.
The focus is in particular on linear terms like
\begin{equation}
 f_{V^{(2)}\tau}({\bf x}_1, {\bf x}_2) := V^{(2)}({\bf x}_1, {\bf x}_2) \, \tau({\bf x}_1, {\bf x}_2),
\label{BG}
\end{equation}
\begin{equation}
 f_{V^{(1)}\tau}({\bf x}_1, {\bf x}_2) := V^{(1)}({\bf x}_1) \, \tau({\bf x}_1, {\bf x}_2),
\label{eff}
\end{equation}
\begin{equation}
 f_{V^{(2)}\circ\tau}({\bf x}_1, {\bf x}_2) := \int V^{(2)}({\bf x}_1, {\bf x}_3) \, \tau({\bf x}_3, {\bf x}_2) \, d{\bf x}_3
\label{oprod}
\end{equation}
and nonlinear terms of the form
\begin{equation}
 f_{\tau\circ V^{(2)}\circ \tau}({\bf x}_1, {\bf x}_2) := \int \tau({\bf x}_1, {\bf x}_3) V^{(2)}({\bf x}_3, {\bf x}_4) \tau({\bf x}_4, {\bf x}_2) \, d{\bf x}_3 d{\bf x}_4
\label{oprod2}
\end{equation}
for a certain asymptotic type of the pair-amplitude $\tau$.
With these definitions at hand, let us briefly outline a suitable
fixed point iteration scheme which illustrates  
some main issues of our approach. Actually, this ansatz for the solution of nonlinear CC type equations represents a canonical  choice in numerical simulations.
The basic structure of our problem can be represented by the
greatly simplified nonlinear equation
\begin{equation}
 A \tau = - V^{(2)}
 -f_{V^{(2)}\tau}
 +f_{V^{(1)}\tau}
 -f_{V^{(2)}\circ\tau}
 -f_{\tau\circ V^{(2)}\circ \tau} ,
\label{Atau}
\end{equation}
where $A$ is an elliptic second order partial differential
operator. A simple iteration scheme for this equation may consist
of the following steps. First solve the equation 
$A \tau_0 = - V^{(2)}$
with fixed right hand side. Calculate $f_{V^{(2)}\tau_0}$,
$f_{V^{(1)}\tau_0}$, $f_{V^{(2)}\circ\tau_0}$ and 
$f_{\tau_0\circ V^{(2)}\circ \tau_0}$ and solve in the
next iteration step  
\begin{equation}
 A \tau_1 = - V^{(2)}
 -f_{V^{(2)}\tau_0}
 +f_{V^{(1)}\tau_0}
 -f_{V^{(2)}\circ\tau_0}
 -f_{\tau_0\circ V^{(2)}\circ \tau_0}.
\label{Atau0}
\end{equation}
The last two steps can be repeated generating a iterative sequence
of linear equations
\begin{equation}
  A \tau_{n+1} = - V^{(2)}
 -f_{V^{(2)}\tau_n}
 +f_{V^{(1)}\tau_n}
 -f_{V^{(2)}\circ\tau_n}
 -f_{\tau_n\circ V^{(2)}\circ \tau_n}.
\label{Ataun}
\end{equation} 
which can be solved in a consecutive manner
until convergence of the sequence $\tau_1,\tau_2,\tau_3\ldots$ has been achieved. Such an iteration scheme is rather convenient for an 
asymptotic analysis of the solutions, see e.g.~\cite{FSS08} where
such an analysis has been actually performed for the nonlinear
Hartree-Fock model. Apparently, the basic problem of such kind of iteration scheme is to state necessary and sufficient conditions for its convergence. Whereas in the case of the Hartree-Fock model
convergence has been proven for certain iteration schemes,
the situation is actually less satisfactory in CC theory, cf.~\cite{RS13,S09}.
In concrete physical and chemical applications, non-convergence
of an iteration scheme can be usually addressed to specific properties of the system under consideration. Therefore let
us simply assume in the following that our iteration scheme
is actually convergent and we focus on the asymptotic properties
of solutions of intermediate steps. 

Within the present work, we are mainly interested in the asymptotic
behaviour of iterated pair-amplitudes $\tau_i$, $i=0,1,\ldots$, near coalescence points of electrons. In order to extract these
properties we apply methods from singular analysis \cite{FHS16} and solve (\ref{Ataun}) via an explicitly constructed
asymptotic parametrix and corresponding Green operator, cf.~\cite{FFHS16} for further details. 
Le us briefly outline the basic idea and and some mathematical techniques from singular analysis involved. 
We do this in a rather informal manner by just mentioning the essential properties of the applied calculus
and refer to the monographs \cite{HS08,Schulze98} for a detailed exposition. 
 In a nutshell, a parametrix $P$
of a differential operator $A$ is a pseudo-differential operator which can be considered as a generalized inverse, i.e., when applied
from the left or right side it yields
\[
 PA = I +G_l \quad \quad AP = I +G_r ,
\]
where the remainders $G_l$ and $G_r$ are left and right Green operators, respectively. In contrast to the standard calculus of pseudo-differential operators on smooth manifolds, our calculus
applies to singular spaces with conical, edge and corner type singularities as well. In the smooth case, remainders correspond to compact operators with smooth kernel function, whereas in the singular calculus Green operators encode important asymptotic information which we want to extract. Acting on (\ref{Ataun}) with the parametrix from the left yields
\begin{equation}
 \tau_{n+1} = - PV^{(2)}
 -Pf_{V^{(2)}\tau_n}
 +Pf_{V^{(1)}\tau_n}
 -Pf_{V^{(2)}\circ\tau_n}
 -Pf_{\tau_n\circ V^{(2)}\circ \tau_n}
 -G_l \tau_{n+1} .
\label{PAtaun}
\end{equation} 
The parametrix $P$ maps in a controlled manner between functions with certain asymptotic behaviour which means that we can derive from the asymptotic properties of the terms on the right hand side
of (\ref{Ataun}) its effect on the asymptotic behaviour of $\tau_{n+1}$. Furthermore it is an essential property of $G_l$ that the operator maps onto a space with specific asymptotic type. Therefore
the asymptotic type of $G_l \tau_{n+1}$ is fixed and does not depend on $\tau_{n+1}$. Thus we have full control on the asymptotic properties of the right hand side of (\ref{PAtaun}) and consequently on the asymptotic type of the iterated pair-amplitude $\tau_{n+1}$.

At this point of our discussion, it is convenient to introduce a diagrammatic notation indispensable in quantum many-particle theory. The following considerations are based on Goldstone diagrams,
cf.~\cite{LM86,NO88} for a comprehensive discussion from a physical point of view. For the mathematically inclined reader let us briefly outline the basic idea. The iteration scheme discussed in the
previous paragraphs can be further decomposed by taking into account the linearity of the differential operator $A$. Instead of solving the first iterated equation (\ref{Atau0}) as a whole, let us
consider the decomposition
\[
 A \tau_{1,1} = -f_{V^{(2)}\tau_0}, \quad A \tau_{1,2} = f_{V^{(1)}\tau_0}, \quad A \tau_{1,3} = -f_{V^{(2)}\circ\tau_0}, \quad
 A \tau_{1,4} = -f_{\tau_0\circ V^{(2)}\circ \tau_0},
\]
from which one recovers the first iterated solution via the sum
\begin{equation}
 \tau_1 = \tau_0 + \tau_{1,1} + \tau_{1,2} + \tau_{1,3} + \tau_{1,4} ,
 \label{t1decomp}
\end{equation}
where each term actually corresponds to an individual Goldstone diagram. In the next iteration step, one can further use the decomposition (\ref{t1decomp}) 
to construct the interaction terms on the right hand side. For each new term on the right hand side obtained in such a manner one can again solve the corresponding equation
which leads to the decomposition of the second iterated solution into Goldstone diagrams
\[
 \tau_2 = \tau_0 + \tau_{1,1} + \tau_{1,2} + \tau_{1,3} + \tau_{1,4} + \tau_{2,1} + \tau_{2,2} + \cdots \ .
\]
This process can be continued through any number of successive iteration steps.
Therefore from a diagrammatic point of view iteration schemes correspond to the summation of an infinite series
of Goldstone diagrams which represents a pair-amplitude. 
Therefore, if we restrict ourselves in the following to study intermediate solutions $\tau_n$, $n \in \mathbb{N}$, within the iteration scheme, we actually consider asymptotic properties of certain finite sums of Goldstone diagrams. 
In particular it is possible to consider specific diagrams
or appropriate subtotals. A possible choice for such a subtotal
of Goldstone diagrams is e.g.~the finite sum of diagrams which represents the progression of an iteration process, i.e.,
\begin{equation}
 P_{n} := \tau_n - \tau_{n-1} \quad \mbox{with} \ n > 0 .
 \label{Pn}
\end{equation}
The regularity and multi-scale features of these subtotals
are especially interesting with respect to the numerical analysis of CC theory. It is e.g.~possible to study their approximation properties with respect to systematic basis sets in appropriate function spaces. 

\section{Asymptotic properties of RPA diagrams}
\label{RPAdiag}
This section contains a summary of our results concerning the asymptotic properties of iterated pair-amplitudes and certain classes of RPA diagrams. For the sake of a reader not interested in the mathematical details 
of the present work it can be read independently from the rest of the paper. The iterated pair-amplitudes can be decomposed into a finite number of Goldstone diagrams each of them has a characteristic asymptotic behaviour 
near coalescence points of electrons. In the following, let us denote by $\tau_{\rpa}$ an arbitrary Goldstone diagram which contributes to an iterated pair-amplitude $\tau^{(n)}$. It is convenient to study 
iterated pair-amplitudes and Goldstone diagrams with respect to the alternative
Cartesian coordinates
\begin{equation}
 {\bf x} := {\bf x}_1, \quad {\bf z} := {\bf x}_1 - {\bf x}_2 ,
\label{xzcoord}
\end{equation}
which become our standard Cartesian coordinates in the remaining part of the paper. By abuse of notation, we refer to iterated pair-amplitudes $\tau^{(n)}$ and Goldstone diagrams $\tau_{\rpa}$
either with respect to $({\bf x}_1,{\bf x}_2)$ or $({\bf x},{\bf z})$ variables, i.e., $\tau_{\rpa}({\bf x}_1,{\bf x}_2) \equiv \tau_{\rpa}({\bf x},{\bf z})$. 
As already mentioned before, the focus of the present work is on asymptotic expansions of iterated pair-amplitudes 
and Goldstone diagrams near coalescence points of electrons. i.e., ${\bf z} \ \rightarrow 0$. This will be achieved by identifying these quantities with kernel functions of classical pseudo-differential 
operators. These operators provide an algebra which enables an efficient treatment of interaction terms, in particular the nonlinear ones as will be demonstrated below. 
Let the corresponding symbol of a Goldstone diagram $\tau_{\rpa}$ be given by
\[
 \sigma_{\rpa}({\bf x},\boldeta) := \int e^{-i {\bf z} \boldeta} \tau_{\rpa}({\bf x},{\bf z}) \, d{\bf z} .
\]
The symbol belongs to the standard H\"ormander class $S^p(\mathbb{R}^3 \times \mathbb{R}^3)$ if it belongs to $C^\infty(\mathbb{R}^3 \times \mathbb{R}^3)$ and satisfies the estimate
\[
 \left| \partial_{\bf x}^\alpha \partial_{\boldeta}^\beta \sigma_{\rpa}({\bf x},\boldeta) \right| \lesssim \bigl( 1+ |\boldeta| \bigr)^{p-|\beta|} \quad \mbox{for all} \ {\bf x},\boldeta \in \mathbb{R}^3 .
\]
Here and in the following $a \lesssim b$ means that $a \leq Cb$ for some constant $C$ which is
independent of variables or parameters on which $a$, $b$ may depend on.
Furthermore, it belongs to the class $S^p_{cl}(\mathbb{R}^3 \times \mathbb{R}^3)$, $p \in \mathbb{Z}$, of classical symbols
if a decomposition
\begin{equation}
 \sigma_{\rpa}({\bf x},\boldeta) = \sum_{j=0}^{N-1} \sigma_{p-j}({\bf x},\boldeta) + \sigma_{p-N}({\bf x},\boldeta) 
\label{asymbcl}
\end{equation}
into symbols $\sigma_{p-j} \in S^{p-j}(\mathbb{R}^3 \times \mathbb{R}^3)$ and remainder $\sigma_{p-N} \in S^{p-N}(\mathbb{R}^3 \times \mathbb{R}^3)$ for any $N \in \mathbb{N}$
exits, such that for $\lambda \geq 1$ and $\eta$ greater some constant, we have $\sigma_{p-j}({\bf x}, \lambda \boldeta) = \lambda^{p-j} \sigma_{p-j}({\bf x},\boldeta)$.
The asymptotic expansion of a classical Goldstone symbol in Fourier space is related to a corresponding asymptotic expansion of its kernel function.
In the following theorem we establish the connection between Goldstone diagrams and classical pseudo-differential operators and give a simple rule to determine
the symbol class $S^p_{cl}(\mathbb{R}^3 \times \mathbb{R}^3)$ to which they belong. Furthermore, the theorem provides the asymptotic expansion of a Goldstone diagram
near coalescence points of electrons.
Let us distinguish in the following between smooth and singular contributions to an asymptotic expansion. Here smooth refers to asymptotic terms which belong to $C^{\infty}(\mathbb{R}^{3} \times \mathbb{R}^{3})$.
Such terms do not cause any difficulties for approximation schemes applied in numerical simulations. Actually, the asymptotic analysis discussed below, does not provide much information concerning smooth terms,
instead it focuses on singular contributions which determine the computational complexity of numerical methods for solving CC equations. 

\begin{theorem}
Goldstone diagrams of RPA-CC pair-amplitudes can be considered as kernel functions of classical pseudo-differential operators without logarithmic terms in their asymptotic expansion. Classical symbols (\ref{asymbcl}) corresponding to Goldstone diagrams belong to the symbol classes $S_{\cl}^{p}$ with $p \leq -4$.
The asymptotic expansion of a Goldstone diagram $\tau_{\rpa}$ with symbol $\sigma_{\rpa} \in S_{\cl}^{p}$, expressed in spherical coordinates $(z,\theta,\phi)$, is given by
\begin{equation}
 \tau_{\rpa}({\bf x},{\bf z}) \sim \sum_{0 \leq j} \tau_{p-j} ({\bf x},z,\theta,\phi) \quad \mbox{modulo} \ C^{\infty}(\mathbb{R}^3\times \mathbb{R}^3) ,
\label{tausing}
\end{equation}
with
\[
 \tau_{p-j}({\bf x},z,\theta,\phi) = z^{j-p-3} \sum_{\substack{l=0\\j-p-l \even}}^{j-p-3} \sum_{m=-l}^l g_{j,lm}({\bf x}) \, Y_{lm}(\theta,\phi)  ,
\]
where functions $g_{j,lm}$ belong to $C^{\infty}(\mathbb{R}^3)$.
In the following, we refer to (\ref{tausing}) as the singular part of the asymptotic expansion of a Goldstone diagram. 

The symbol class of a diagram $\tau_{\rpa}$ can be determined in the following manner
\begin{itemize}
\item[i)] Remove all ladder insertions in the diagram.
\item[ii)] Count the number of remaining interaction lines $n$.
\end{itemize}
Then the corresponding symbol of the diagram $\tau_{\rpa}$ belongs to the symbol class $S_{\cl}^{-4n}$, cf.~Fig.~\ref{fig3}. 
\label{theorem1}
\end{theorem}

An appropriate measure of the singular behaviour of Goldstone diagrams
is the so-called asymptotic smoothness property discussed in the following corollary, to which we refer in Section \ref{Besovregularity} where we discuss
approximation properties of these diagrams.
\begin{corollary}
A Goldstone diagram $\tau_{\rpa}$, with corresponding symbol in the symbol class $S_{\cl}^{p}$, belongs to $C^{\infty}(\mathbb{R}^3\times \mathbb{R}^3 \setminus \{ 0 \})$ and satisfies the asymptotic smoothness property
\begin{equation}
\left| \partial_{\bf x}^{\beta} \partial_{\bf z}^{\alpha}
\tau_{\rpa} \right| \lesssim |{\bf z}|^{-3-p-|\alpha|-N}
\quad \mbox{for} \  -3-p-|\alpha|-N  <0, \  \mbox{and any} \ N \in \mathbb{N}_0,
\label{asympestimate}
\end{equation}
where for $|\alpha| \leq -3-p$ it has bounded partial derivatives.
\label{corollaryasymp}
\end{corollary}

\begin{figure}[t]
\begin{center}
\includegraphics[scale=0.3]{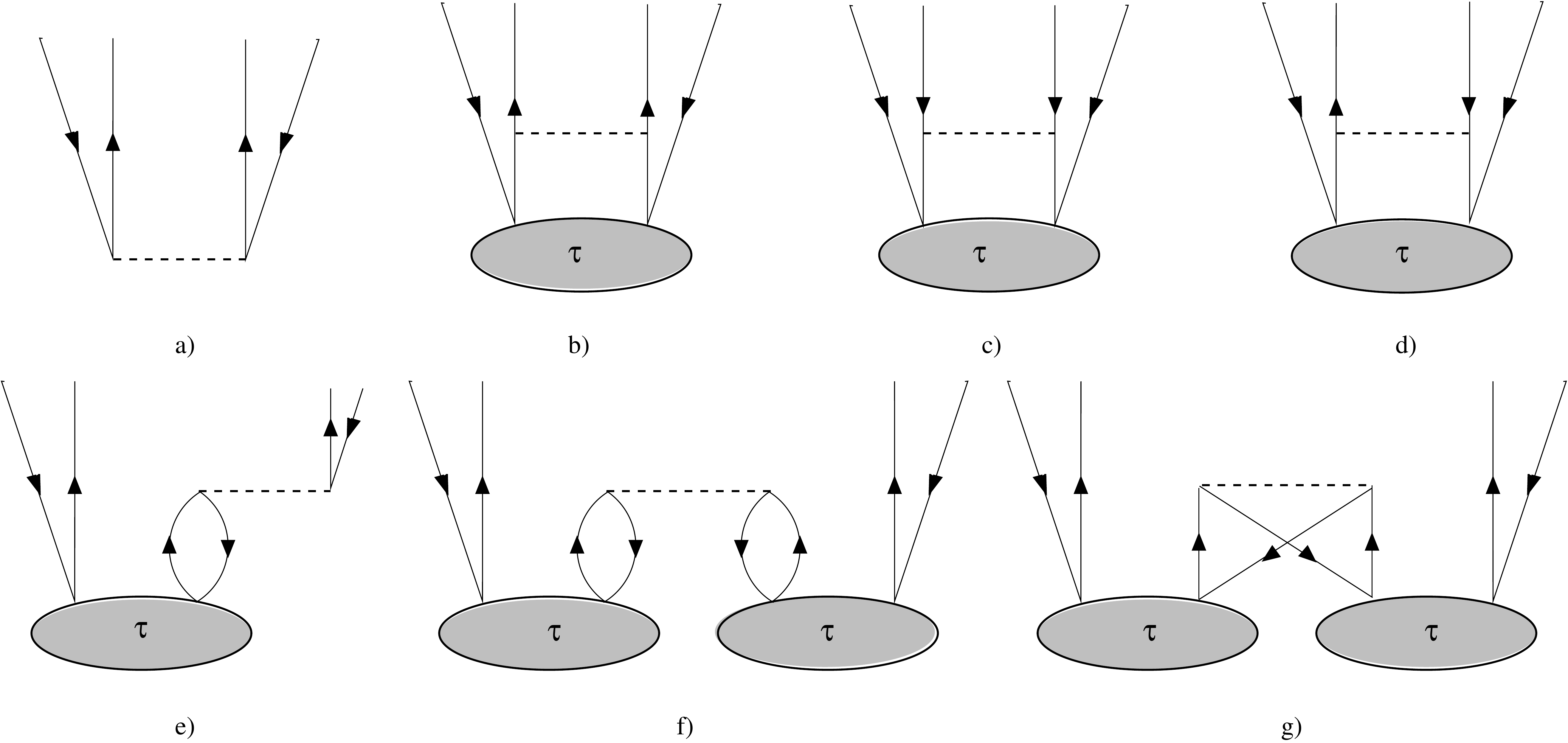}
\end{center}
\vspace{0.5cm}
\caption{Goldstone diagrams of the SUB-2 terms taken into account by our model. a) Coulomb interaction (\ref{MP1}). b) - d) Particle, hole and particle-hole ladder diagrams which contribute to (\ref{ladder}) and (\ref{RPA1}).
e) Linear RPA diagram which contributes to (\ref{RPA2}). f), g) Nonlinear RPA diagrams, corresponding to (\ref{RPA3}), direct and exchange contributions, respectively.}
\label{fig2}
\end{figure}

\begin{figure}[t]
\begin{center}
\includegraphics[scale=0.25]{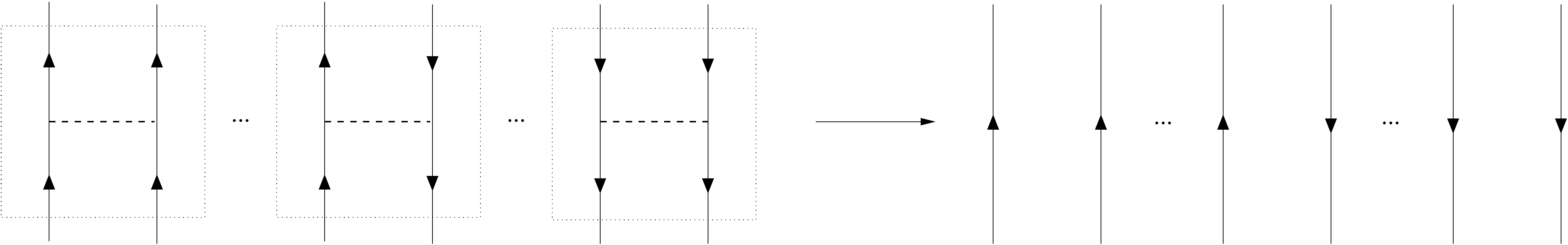}\\ \vspace{1.cm}
\includegraphics[scale=0.23]{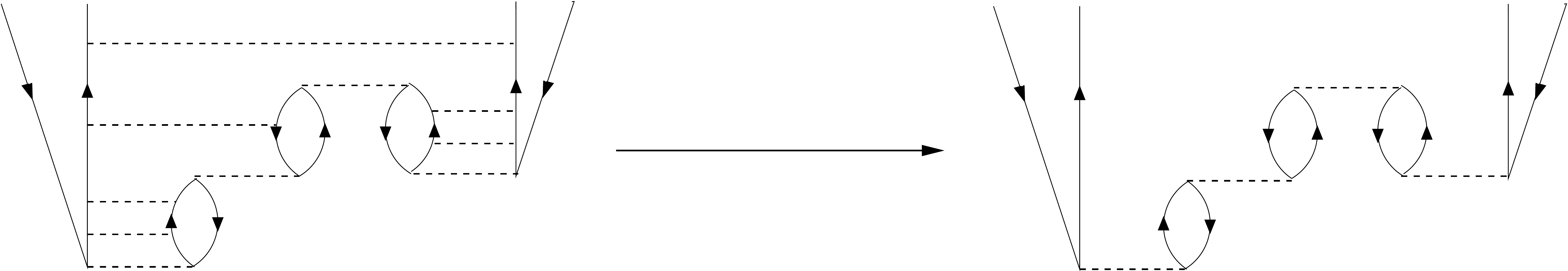}
\end{center}
\vspace{0.5cm}
\caption{Stripping of ladder contributions from RPA Goldstone diagrams. Upper part: cutting-out of particle, hole and particle-hole ladder insertions. Lower part: example of a RPA diagram with multiple ladder insertions. The symbol of the resulting
forth-order RPA diagram belongs to the symbol class $S_{\cl}^{-16}$ as well as the symbol corresponding to the original diagram.}
\label{fig3}
\end{figure}

According to Theorem \ref{theorem1}, the ladder diagrams b), c) and d) in Figure \ref{fig2} do not alter symbol classes within the standard iteration scheme outlined in Section \ref{iterationscheme}. 
In the standard RPA models usually considered in the literature, cf.~\cite{SHS08}, these diagrams are neglected altogether. What remains are the RPA diagrams e), f) and g) in Figure \ref{fig2}, which represent the driving terms of the iteration scheme.

\begin{corollary}
The symbols of Goldstone diagrams representing the progression $P_n$, cf.~(\ref{Pn}), of the $n$'th iteration step of standard RPA models, i.e.~no ladder insertions, can be classified according to the descending filtration of symbol classes 
\[
 S_{\cl}^{-4(n+1)} \supset S_{\cl}^{-4(n+2)} \supset S_{\cl}^{-4(n+3)} \supset \cdots \supset S_{\cl}^{-4(2^{n+1}-1)} .
\]
\label{corollaryfiltration}
\end{corollary}
 
 The following three sections provide the individual building blocks required for the proof of Theorem \ref{tausing} and its corollaries. These building blocks are finally put together in Section \ref{proofmaintheorem} for the proof of the theorem.
 
\section{Pair wavefunctions, classical pseudo-differential operators and wedge Sobolev spaces}
\label{psiDO}
It is the purpose of this section to develop some technical tools in order to deal with the asymptotic behaviour of certain linear and nonlinear terms, cf.~Section \ref{iterationscheme}, within the RPA-CC model.
Once again, indices referring to occupied orbitals and spin degrees of freedom have been omitted.

\subsection{Kernel functions and wedge Sobolev spaces}
\label{kernel-wedge}
In the present work, we follow a dual approach where we either
consider iterates of pair-amplitudes as functions in so-called wedge Sobolev spaces with asymptotics, to be defined below, or as kernel functions of classical pseudo-differential operators. At first it must be shown that the initial iterate 
$\tau_0$ actually fits into this setting. Thereafter all linear
and nonlinear interaction terms, cf.~(\ref{BG}), (\ref{eff}), 
(\ref{oprod}) and (\ref{oprod2}), are studied and their consistency
with our setting has to be proven. Under the hypothesis that such kind of
approach is actually feasible let us outline the underlying concept
in some detail. 

From the point of view of singular analysis it is convenient to
consider  the configuration space $\mathbb{R}^6$ of two electrons
as a stratified manifold with embedded edge and corner singularities, cf.~\cite{FHS15}.
For this purpose let us introduce hyperspherical coordinates in $\mathbb{R}^6$ with radial variable
\begin{equation}
 t:= \sqrt{x_1^2 + x_2^2 + x_3^2 + x_4^2 + x_5^2 + x_6^2} .
\label{t-def}
\end{equation}
In these coordinates, $\mathbb{R}^6$ can be formally considered as a conical manifold
with compact base $B$, homeomorphic to $S^5$, and embedded conical singularity at the origin, i.e.,
\[
 \mathbb{R}^6 \equiv (S^5)^\Delta:=(\overline{\mathbb{R}}_+ \times S^5) / (\{0\} \times S^5) .
\]
Here, the origin formally represents a higher order corner singularity, cf.~\cite{ES97} for further details concerning manifolds with singularities.
Due to the absence of singular electron-nuclear interactions, however, there is no corresponding 
physical higher order singularity at the origin, instead it belongs to the edge corresponding to
coalescence points of electrons. 
Removal of this ``singular'' point defines an open stretched cone
\[
 (S^5)^{\wedge} := \mathbb{R}_+ \times S^5 .
\]
In order to avoid possible complications, let us
represent the configuration space
of two electrons in the following by a ``hyperspherical atlas'' of at least two open sets with local hyperspherical coordinates such that $t > c$, for an appropriate constant $c>0$, is guaranteed with respect to
any local chart. Let us briefly discuss the singular structure on the base of the cone $S^5$ and refer to \cite{FHS15} for further details.
On the base of the cone, we have a closed embedded submanifold $Y$ which represents the singular edge of coalescence points of electrons.
The submanifold $Y$ is homeomorphic to $S^2$ and there exists a local neighbourhood $U$ on $S^5$ which is homeomorphic to a wedge
\[
 W = X^{\Delta} \times Y \ \ \mbox{with} \ X^{\Delta} :=
 (\overline{\mathbb{R}}_+ \times X) / (\{0\} \times X),
\]
where the base $X$ of the wedge is again homeomorphic to $S^2$.
The corresponding open stretched wedge is
\[
 \mathbb{W} = X^{\wedge} \times Y \ \ \mbox{with} \ X^{\wedge} :=
\mathbb{R}_+ \times X .
\]

The hyperspherical coordinates associated to a wedge are defined with respect to 
center of mass coordinates ${\bf z}_1 :=(z_1,z_2,z_3)$
and ${\bf z}_2 :=(z_4,z_5,z_6)$ with
\begin{gather*}
 z_1 = \frac{1}{\sqrt{2}} (x_1 - x_4), \ \ z_2 = \frac{1}{\sqrt{2}} (x_2 - x_5), \ \
 z_3 = \frac{1}{\sqrt{2}} (x_3 - x_6), \\
 z_4 = \frac{1}{\sqrt{2}} (x_1 + x_4), \ \ z_5 = \frac{1}{\sqrt{2}} (x_2 + x_5), \ \
 z_6 = \frac{1}{\sqrt{2}} (x_3 + x_6)
 \end{gather*}
and explicitly given by
\[
 z_1 = t \sin\,r \sin \theta_1 \cos \phi_1, \ \ z_2 = t \sin\,r \sin \theta_1
\sin \phi_1, \ \
 z_3 = t \sin\,r \cos \theta_1 ,
\]
\[
 z_4 = t \cos\,r \sin \theta_2 \cos \phi_2, \ \ z_5 = t \cos\,r \sin \theta_2
\sin \phi_2, \ \
 z_6 = t \cos\,r \cos \theta_2 ,
\]
with $ t \in (0,\infty)$, $r \in (0,\frac{\pi}{2}]$, $\theta_1,  \theta_2 \in(0,\pi)$, $\phi_1, \phi_2 \in[0, 2\pi)$, cf.~\cite{FFHS16} for further details. The center of mass coordinates are related to our standard Cartesian coordinates via
\begin{equation}
 {\bf x} := \tfrac{1}{\sqrt{2}} \bigl( {\bf z}_1 + {\bf z}_2 \bigr) , \quad {\bf z} := \sqrt{2} {\bf z}_1 .
\label{xzhyper}
\end{equation}

According to our previous discussion, the hyperspherical radius
$t$ does not belong to the edge variables. Instead it should be assigned to the corner singularity as a distance variable and it
is possible to specify, within the corner degenerate pseudo-differential calculus, the $t$-asymptotic behaviour of solutions 
near $t=0$. Due to the absence of a physical corner singularity at this point, we are actually not interested in the $t$-asymptotic behaviour
of solutions. Therefore it is reasonable to treat $t$ simply as yet another edge variable which makes perfectly sense if we want to study the $r$-asymptotic behaviour of solutions locally near coalescence points of electrons, cf.~\cite{FFHS16} for further details. However within the present work we also have to take care
of the exit behaviour of solutions for $t \rightarrow \infty$.
The latter is a natural consequence of the corner degenerate
calculus, cf.~\cite{CS16}. In the following let us apply the corner degenerate
calculus in a rather formal manner to keep control of the exit behaviour and stick to the simpler edge-degenerate calculus for actual calculations of the $r$-asymptotic behaviour of solutions. 
In order to avoid possible confusion let us denote in the following the edge by ${\cal Y}$ $(Y)$ in the corner (edge) degenerate case.
With respect to hyperspherical coordinates the edges ${\cal Y}$ and $Y$ are represented by the coordinates $(\theta_2, \phi_2)$ and $(t,\theta_2, \phi_2)$, respectively. 

It is a consequence of ellipticity theory in the corner degenerate case \cite{CS16,FFHS17} and of the discussion of nonlinear interaction terms below, that iterated pair-amplitudes belong to a Schwartz-type corner space, cf.~\cite{CS16}, 
\[
 {\cal S}^{\gamma_2} \bigl( \mathbb{R}_+,H^{\infty,\gamma_1}(B) \bigr) := \omega_2 {\cal H}^{\infty,\gamma_2} \bigl( \mathbb{R}_+,H^{\infty,\gamma_1}(B) \bigr) + \left. (1-\omega_2) {\cal S}\bigl( \mathbb{R},H^{\infty,\gamma_1}(B) \bigr) \right|_{\mathbb{R}_+} ,
\]
where $\omega_2 \in C_0^\infty(\overline{\mathbb{R}}_+)$ corresponds to a cut-off function which equals 1 on a given interval $[0,\epsilon)$, $\epsilon >0$. 
The range of weights $\gamma_1,\gamma_2$ appropriate for our application have been studied in \cite{FH10,FFH16b}.
In the present work, we do not consider the corner singularity, and by an appropriate choice of charts within the atlas,
pair-amplitudes belong to
\begin{equation}
{\cal S} \bigl( \mathbb{R}_+,H^{\infty,\gamma_1}(B) \bigr) := \left. (1-\omega_2) {\cal S}\bigl( \mathbb{R},H^{\infty,\gamma_1}(B) \bigr) \right|_{\mathbb{R}_+} .
\label{Schwartztype}
\end{equation}
Therefore, the remaining part of the Schwartz-type corner space does not enter into our discussion and we refer to \cite{CS16} for its definition and properties.
The Schwartz space ${\cal S}\bigl( \mathbb{R},H^{\infty,\gamma_1}(B) \bigr)$ consists of all functions $u \in C^{\infty}\bigl(\mathbb{R},H^{\infty,\gamma_1}(B) \bigr)$ such that
\[
 \sup_{t \in \mathbb{R}} \| P(t) Q(\partial_t) u \|_{H^{s,\gamma_1}} < \infty , \quad \forall \, s \in \mathbb{R} ,
\]
for any polynomials $P,Q$, cf.~\cite{Schwartz}. The bound is given with respect to the system of norms 
of the Fr\'echet space $H^{\infty,\gamma_1}(B) := \cap_s H^{s,\gamma_1}(B)$, where the intersection is taken with respect to the edge Sobolev spaces
\[
 H^{s,\gamma_1}(B) := \omega_1 {\cal W}^{s}_{\loc} \bigl( {\cal Y},{\cal K}^{s,\gamma_1}_{P_1}  (X^{\wedge}) \bigr) + (1-\omega_1) H^{s}_{\loc} \bigl( S_0(B) \bigr) ,
\]
defined on the compact base $B$ of the cone which carries the edge-type singularity. Here $\omega_1$ is another cut-off function, $S_0(B)$ denotes the smooth interior of the base $B$ and ${\cal W}^{s}_{\loc} \bigl( {\cal Y},{\cal K}^{s,\gamma_1}_{P_1} \bigr)$
is a weighted wedge Sobolev space with asymptotics. For the ease of the reader, we have summarized in Appendix \ref{appendix1} some basic definitions and properties of these spaces.

Let us consider an atlas on the configuration space of the electrons where individual charts represent the electron pair in hyperspherical coordinates. For a given chart let us represent
the pair-amplitude in hyperspherical coordinates and define
\[
 \tilde \tau(r,\theta_1,\phi_1,t,\theta_2, \phi_2) := 
 \varphi(\theta_2,\phi_2) \tau(r,\theta_1,\phi_1,t,\theta_2, \phi_2) ,
\]
where $\varphi \in C^\infty_0({\cal Y})$ belongs to an appropriate partition of unity on the edge ${\cal Y} \sim S^{2}$.
The pair-amplitude can be decomposed via a cut-off function $\omega$ into a singular edge and smooth inner part in the following way
\begin{equation}
\tilde{\tau}(r,\theta_1,\phi_1,t,\theta_2, \phi_2) = \omega(r) \tilde{\tau}(r,\theta_1,\phi_1,t,\theta_2,\phi_2) +
 \bigl( 1- \omega(r) \bigr) \tilde{\tau}(r,\theta_1,\phi_1,t,\theta_2,\phi_2) .
\label{tdecomp}
\end{equation}

Let us suppose, that the pair-amplitude near an edge 
\begin{equation}
 \omega(r) \tilde{\tau}(r,\theta_1,\phi_1,t,\theta_2,\phi_2) , 
\label{tedge}
\end{equation}
belongs to ${\cal S}\bigl( \mathbb{R}_+,\omega {\cal W}_{\comp}^\infty \bigl( {\cal Y}, {\cal K}_{P}^{\infty, \gamma}(X^\wedge) \bigr) \bigr) \subset {\cal S} \bigl( \mathbb{R}_+,H^{\infty,\gamma}(B) \bigr)$ 
for weight data ${\bf g} =(\gamma, \Theta)$ with
explicit asymptotic expansion
\begin{equation}
 \omega \tilde{\tau} \sim \omega(r) \sum_{j=0}^{N_\Theta} \sum_{k=0}^{m_j} r^{-p_j} \log^k r \, c_{jk}(\theta_1,\phi_1) v_{jk}(t,\theta_2,\phi_2) + h_{\Theta}(r,\theta_1,\phi_1,t,\theta_2,\phi_2)
\label{asymp1}
\end{equation}
where $c_{jk} \in C^\infty(X)$, $v_{jk} \in {\cal S}\bigl( \mathbb{R}_+,H^\infty_{\comp}({\cal Y}) \bigr)$ and 
$h_{\Theta} \in {\cal S}\bigl( \mathbb{R}_+,\omega {\cal W}_{comp}^\infty \bigl( {\cal Y}, {\cal K}_{\Theta}^{\infty, \gamma}(X^\wedge) \bigr)) \bigr)$, cf.~(\ref{appasymp}) in Appendix \ref{appendix1}.
Furthermore let us specify the weight data according to $\frac{1}{2} < \gamma < \frac{3}{2}$, cf.~\cite{FH10} with
\[
 \tfrac{3}{2} -\gamma + \vartheta < \Re p_j < \tfrac{3}{2} -\gamma
\]
where in particular the Taylor asymptotic type $-p_j = 0,1,2,\cdots$ with $m_j= 0$ ($j$ odd) and $m_j=0,1$ ($j$ even) will be assumed.
In the following we refer to this case when we mention generic weight data or generic asymptotic type.
According to our assumptions, the complementary part of the pair-amplitude
\begin{equation}
 \bigl( 1- \omega(r) \bigr) \tilde{\tau}(r,\theta_1,\phi_1,t,\theta_2,\phi_2) ,
\label{tcompl}
\end{equation}
belongs to a Schwartz-type corner space ${\cal S} \bigl( \mathbb{R}_+,H^{\infty,\gamma_1}(B) \bigr)$. For an appropriate choice of $\omega$, we can further restrict (\ref{tcompl}) to the inner part
$(1-\omega_1) H^{\infty}_{\loc} \bigl( S_0(B) \bigr)$.

As already mentioned before, we consider a dual approach where
the pair-amplitude also represents a kernel function of a classical pseudo-differential operator. In order to express this duality let
us first define partial derivatives with respect to hyperspherical coordinates 
\[
 \partial_{X^\wedge} := \begin{pmatrix} -r \partial_r \\ \partial_{\theta_1} \\ \partial_{\phi_1} \end{pmatrix} \quad \quad
 \partial_{Y} := \begin{pmatrix} \partial_t \\ \partial_{\theta_2} \\ \partial_{\phi_2} \end{pmatrix}
\]
and appropriate Cartesian coordinates, cf.~(\ref{xzcoord}),
\[
 \partial_{{\bf x}} := \begin{pmatrix} \partial_{x_1} \\ \partial_{x_2} \\ \partial_{x_3} \end{pmatrix} \quad \quad
 \partial_{{\bf z}} := \begin{pmatrix} \partial_{z_1} \\ \partial_{z_2} \\ \partial_{z_3} \end{pmatrix} .
\]
These derivatives are related via the linear transformation
\[
 \begin{pmatrix} \partial_{{\bf x}} \\ \partial_{{\bf z}} \end{pmatrix} = \mathbb{T}
 \begin{pmatrix} \partial_{X^\wedge} \\ \partial_{Y} \end{pmatrix}
\]
given by the matrix
\[
 \mathbb{T} = \begin{pmatrix} \begin{array}{c} \vdots \\ -\frac{\sqrt{2}}{r} \partial_{{\bf z}_2} r \\ \vdots \end{array} & 
                 \begin{array}{c} \vdots \\ 0 \\ \vdots \end{array} &
                 \begin{array}{c} \vdots \\ 0 \\ \vdots \end{array} &
                 \begin{array}{c} \vdots \\ \sqrt{2} \partial_{{\bf z}_2} t \\ \vdots \end{array} &
                 \begin{array}{c} \vdots \\ \sqrt{2} \partial_{{\bf z}_2} \theta_2 \\ \vdots \end{array} &
                 \begin{array}{c} \vdots \\ \sqrt{2} \partial_{{\bf z}_2} \phi_2 \\ \vdots \end{array} \\
                 \begin{array}{c} \vdots \\ \frac{1}{\sqrt{2}r} \bigl( \partial_{{\bf z}_2} r - \partial_{{\bf z}_1} r \bigr) \\ \vdots \end{array} &
                 \begin{array}{c} \vdots \\ \tfrac{1}{\sqrt{2}} \partial_{{\bf z}_1} \theta_1 \\ \vdots \end{array} &
                 \begin{array}{c} \vdots \\ \tfrac{1}{\sqrt{2}} \partial_{{\bf z}_1} \phi_1 \\ \vdots \end{array} &
                 \begin{array}{c} \vdots \\ \tfrac{1}{\sqrt{2}} \bigl( \partial_{{\bf z}_1} t - \partial_{{\bf z}_2} t \bigr) \\ \vdots \end{array} &
                 \begin{array}{c} \vdots \\ -\tfrac{1}{\sqrt{2}} \partial_{{\bf z}_2} \theta_2 \\ \vdots \end{array} &
                 \begin{array}{c} \vdots \\ -\tfrac{1}{\sqrt{2}} \partial_{{\bf z}_2} \phi_2 \\ \vdots \end{array} \end{pmatrix} ,
\]
with
\begin{align*}
 \partial_{{\bf z}_1} r & = \frac{|{\bf z}_2|}{|{\bf z}_1| \bigl( |{\bf z}_1|^2 + |{\bf z}_2|^2 \bigr)} \begin{pmatrix} z_1 \\ z_2 \\ z_3 \end{pmatrix} &
 \partial_{{\bf z}_2} r & = -\frac{|{\bf z}_1|}{|{\bf z}_2| \bigl( |{\bf z}_1|^2 + |{\bf z}_2|^2 \bigr)} \begin{pmatrix} z_4 \\ z_5 \\ z_6 \end{pmatrix} \\
 \partial_{{\bf z}_1} \theta_1 & = -\frac{1}{|{\bf z}_1|^2 \sqrt{z_1^2 + z_2^2}} \begin{pmatrix} -z_1z_3 \\ -z_2z_3 \\ z_1^2+z_2^2 \end{pmatrix} &
 \partial_{{\bf z}_1} \phi_1 & = \frac{1}{z_1^2 + z_2^2} \begin{pmatrix} -z_2 \\ z_1 \\ 0 \end{pmatrix} \\
 \partial_{{\bf z}_1} t & = \frac{1}{\sqrt{|{\bf z}_1|^2 + |{\bf z}_2|^2}} \begin{pmatrix} z_1 \\ z_2 \\ z_3 \end{pmatrix} &
 \partial_{{\bf z}_2} t & = \frac{1}{\sqrt{|{\bf z}_1|^2 + |{\bf z}_2|^2}} \begin{pmatrix} z_4 \\ z_5 \\ z_6 \end{pmatrix} \\
 \partial_{{\bf z}_2} \theta_2 & = -\frac{1}{|{\bf z}_2|^2 \sqrt{z_4^2 + z_5^2}} \begin{pmatrix} -z_4z_6 \\ -z_5z_6 \\ z_4^2+z_5^2 \end{pmatrix} &
 \partial_{{\bf z}_2} \phi_2 & = \frac{1}{z_4^2 + z_5^2} \begin{pmatrix} -z_5 \\ z_4 \\ 0 \end{pmatrix} 
\end{align*}
Let us recall the following well known relation between kernel functions and symbols of pseudo-differential operators, cf.~\cite{Stein}.
\begin{proposition}
\label{SteinPr}
Given a pseudo-differential operator $T$ with symbol $\sigma \in S^{m}(\mathbb{R}^3 \times \mathbb{R}^3)$ for $m<0$. The kernel function
$k ({\bf x}, {\bf z}) \in C^\infty(\mathbb{R}^3 \times \mathbb{R}^3 \setminus \{0\})$
of the corresponding integral operator satisfies the estimate
\begin{equation}
 |\partial_{{\bf x}}^\beta \partial_{{\bf z}}^\alpha k({\bf x}, {\bf z})|
 \leq C_{\alpha,\beta,N} |{\bf z}|^{-3-m-|\alpha|-N} ,
\label{bound}
\end{equation}
for all multi indices $\alpha, \beta$ and $N \geq 0$ such that $3+m+|\alpha|+N > 0$.

Vice versa, there exists a pseudo-differential operator $T$ with symbol $\sigma \in S^{m}(\mathbb{R}^3 \times \mathbb{R}^3)$ for each
kernel function $k ({\bf x}, {\bf z}) \in C^\infty(\mathbb{R}^n \times \mathbb{R}^n/ \{0\})$,
which satisfy the estimate (\ref{bound}) for $m < 0$, such that
\begin{equation}
 T f({\bf x}) = \int k ({\bf x}, {\bf x} - {\bf y}) \, f({\bf y}) \, d{\bf z} ,
 \ \ \mbox{for all} \ f \in {\cal S} .
\label{Tfkf}
\end{equation}
\end{proposition}

\begin{remark}
In order to determine the asymptotic type of iterative solutions and Goldstone diagrams, we have 
to show that kernel functions locally correspond to functions in an edge degenerate Sobolev space with asymptotics 
${\cal W}_{\loc}^\infty(Y, {\cal K}_{P}^{\infty, \gamma}(X^\wedge))$ and globally to functions in a corner degenerate Sobolev space 
${\cal S} \bigl( \mathbb{R}_+,H^{\infty,\gamma_1}(B) \bigr)$ which characterizes the exit behaviour.
The latter property is not part of the standard pseudo-differential calculus and has to be imposed as an additional property, i.e., we consider in the following
kernel functions satisfying the stronger estimate
\begin{equation}
 |{\bf x}^{\gamma}\partial_{{\bf x}}^\beta \partial_{{\bf z}}^\alpha k({\bf x}, {\bf z})|
 \leq C_{\alpha,\beta,N} |{\bf z}|^{-3-m-|\alpha|-N} ,
\label{xbound}
\end{equation}
for all multi indices $\alpha, \beta,\gamma$ and $N \geq 0$ such that $3+m+|\alpha|+N > 0$. Pseudo-differential operators
whose kernels satisfy (\ref{xbound}) actually form a subalgebra within the pseudo-differential algebra.
\end{remark}

In the following, we apply Proposition \ref{SteinPr} to wedge Sobolev spaces with asymptotics.

\begin{proposition}
\label{asympkernel}
Given a bounded kernel function $k ({\bf x}, {\bf z}) \in C^\infty(\mathbb{R}^n \times \mathbb{R}^n/ \{0\})$.
Let us assume its representative in hyperspherical coordinates $\tilde{k}$ belongs to 
${\cal S} \bigl(\mathbb{R}_+,\omega {\cal W}_{\comp}^\infty \bigl( Y, {\cal K}_{P}^{\infty, \gamma}(X^\wedge) \bigr) \bigr)$ and is of generic asymptotic type.
There exists a pseudo-differential operator $T$ with symbol $\sigma \in S^{-3}(\mathbb{R}^3 \times \mathbb{R}^3)$ for the kernel function $k$ such that
(\ref{Tfkf}) is satisfied.
\end{proposition}
\begin{proof}
Expressing derivatives with respect to hyperspherical coordinates, we get the estimates
\begin{eqnarray*}
 \left| \partial_{{\bf x}}^\beta \partial_{{\bf z}}^\alpha k({\bf x}, {\bf z}) \right| & \lesssim & 
 \sum_{|\alpha_1|+|\alpha_2|=|\alpha|} \sum_{|\beta_1|+|\beta_2|=|\beta|} |{\bf z}|^{-|\alpha_1|}
 \left| \partial^{\alpha_1+\beta_1}_{X^\wedge} \partial^{\alpha_2+\beta_2}_{Y} \tilde{k}(r,\theta_1,\phi_1,t,\theta_2,\phi_2) \right| \\
 & \lesssim & |{\bf z}|^{-|\alpha|} \sum_{|\alpha_1|+|\alpha_2|=|\alpha|} \sum_{|\beta_1|+|\beta_2|=|\beta|} 
 \left| \partial^{\alpha_1+\beta_1}_{X^\wedge} \partial^{\alpha_2+\beta_2}_{Y} \tilde{k}(r,\theta_1,\phi_1,t,\theta_2,\phi_2) \right| ,
\end{eqnarray*}
where $({\bf x}, {\bf z})$ and $(r,\theta_1,\phi_1,t,\theta_2,\phi_2)$ are related via (\ref{xzhyper}).

According to Proposition \ref{SteinPr}, the integral operator with kernel function $k$ corresponds to a pseudo-differential operator $T$ with symbol 
$\sigma \in S^{-3}(\mathbb{R}^3 \times \mathbb{R}^3)$
provided $\partial^{\alpha}_{X^\wedge} \partial^{\beta}_{Y} \tilde{k}$ for any $\alpha$ and $\beta$ belongs to $L^\infty(X^\wedge \times Y)$.
Let us consider for $\tilde{k}$ the asymptotic expansion (\ref{asymp1}). It is obvious from the definition that every term in the asymptotic part has bounded partial derivatives
$\partial^{\alpha}_{X^\wedge} \partial^{\beta}_{Y}$ for any $\alpha$ and $\beta$. Therefore it is sufficient to consider the remainder $\tilde{k}_{\Theta}$ which 
belongs to ${\cal S} \bigl(\mathbb{R}_+,\omega {\cal W}_{\comp}^\infty \bigl( Y, {\cal K}_{\Theta}^{\infty, \gamma}(X^\wedge)\bigr)\bigr)$ with $\Theta =(\vartheta,0]$. 
It follows from a standard Sobolev embedding theorem that $\partial^{\alpha}_{X^\wedge} \partial^{\beta}_{Y} \tilde{k} \in L^\infty(X^\wedge \times Y)$
if it belongs to the Sobolev space $H^4(X^\wedge \times Y)$. Therefore it is sufficient to show that $\partial^k_r \partial^{\alpha}_{X^\wedge} \partial^{\beta}_{Y} \tilde{k}$,
for $k=0,1,\ldots,4$, is square integrable for any multi indices $\alpha$ and $\beta$. Let us assume w.l.o.g.~$\vartheta < - \frac{17}{2}$, it follows
\begin{eqnarray*}
 \sum_{k=0}^4 \int_{X^\wedge} \int_{Y} \left| \partial^k_r \partial^{\alpha}_{X^\wedge} \partial^{\beta}_{Y} \tilde{k} \right|^2 \, dr dx dy 
 & \lesssim & \sum_{k=0}^4 \int_{X^\wedge} \int_{Y} r^{2-2\tilde{\gamma}} \left| (-r \partial_r)^k \partial^{\alpha}_{X^\wedge} \partial^{\beta}_{Y} \tilde{k} \right|^2 \, dr dx dy \\
 & < & \infty 
\end{eqnarray*}
because according to our assumptions we can take $\tilde{\gamma} >9$ and $\tilde{k}_{\Theta}$ belongs to ${\cal S} \bigl(\mathbb{R}_+,\omega {\cal W}_{\comp}^\infty \bigl( Y, {\cal K}^{\infty, \tilde{\gamma}}(X^\wedge)\bigr) \bigr)$.
\end{proof}

\begin{remark}
\label{c00const}
Let us further assume $c_{00}=$ const.~in the leading order term of the asymptotic expansion (\ref{asymp1}). Hence its corresponding symbol belongs to $S^{-\infty}(\mathbb{R}^3 \times \mathbb{R}^3)$
and $\sigma$ belongs to $S^{-4}(\mathbb{R}^3 \times \mathbb{R}^3)$. Actually if we restrict the asymptotic expansion (\ref{asymp1}) to $j \geq j_0 >0$ with $j_0$ odd, one gets
\begin{eqnarray*}
 \lefteqn{\tilde{k} \sim \omega(r) \sum_{j \geq j_0}^{N_\Theta} \sum_{k=0}^{m_j} r^{-p_j} \log^k r \, c_{jk}(\theta_1,\phi_1) v_{jk}(t,\theta_2,\phi_2) + \tilde{k}_{\Theta}(r,\theta_1,\phi_1,t,\theta_2,\phi_2) =} \\
 & & r^{j_0-1} \, \left[ \omega(r) \sum_{j \geq j_0}^{N_\Theta} \sum_{k=0}^{m_j} r^{-p_j-j_0+1} \log^k r \, c_{jk}(\theta_1,\phi_1) v_{jk}(t,\theta_2,\phi_2) + r^{-j_0+1} \tilde{k}_{\Theta}(r,\theta_1,\phi_1,t,\theta_2,\phi_2) \right] ,
\end{eqnarray*}
where it follows, using the same arguments as before, that the restricted kernel function corresponds to a symbol in $S^{-3-j_0}(\mathbb{R}^3 \times \mathbb{R}^3)$.
\end{remark}

With respect to our particular application, Proposition \ref{asympkernel} refers to kernel functions which actually correspond to the part of a pair-amplitude located in a neigbourhood of the edge, cf.~(\ref{tedge}).
Linear and nonlinear interaction terms, cf.~(\ref{BG}), (\ref{eff}), (\ref{oprod}) and (\ref{oprod2}), however, depend on the whole pair-amplitude $\tau({\bf x}_1, {\bf x}_2)$, ${\bf x}_1, {\bf x}_2 \in \mathbb{R}^3$. 
Therefore, let us also consider the complementary part (\ref{tcompl})
which according to our assumptions belongs to a Schwartz-type corner space.

\begin{proposition}
\label{Schwartz}
The corresponding symbol of the kernel function (\ref{tcompl}) which represents the complementary part of a pair-amplitude belongs to $S^{-\infty}(\mathbb{R}^3 \times \mathbb{R}^3)$.
\end{proposition}
\begin{proof}
From our previous considerations, we can derive the following estimate
\begin{eqnarray*}
 \left| \partial^{\alpha}_{X^\wedge} \partial^{\beta}_{Y} (1- \omega) \tilde{\tau} \right|
 & \lesssim & |{\bf z}|^{-N} \left| t^N \partial_t^{\beta_1} \partial^{\alpha}_{X^\wedge} \partial^{\beta_2}_{\theta_2,\phi_2}(1- \omega) \tilde{\tau} \right| \\
 & \lesssim & |{\bf z}|^{-N} \sup_{t > c} \left\| t^N \partial_t^{\beta_1} \bigl( \partial^{\alpha}_{X^\wedge} \partial^{\beta_2}_{\theta_2,\phi_2}(1- \omega) \tilde{\tau} \bigr)
 \right\|_{H^{m}_{\loc} \bigl( S_0(B) \bigr)} \quad (m > \tfrac{5}{2}) \\
 & \lesssim & |{\bf z}|^{-N} \sup_{t > c} \left\| t^N \partial_t^{\beta_1} \bigl( \partial^{\alpha}_{X^\wedge} \partial^{\beta_2}_{\theta_2,\phi_2}(1- \omega) \tilde{\tau} \bigr)
 \right\|_{H^{m,\gamma_1}(B)} \\
 & \lesssim & |{\bf z}|^{-N} ,
\end{eqnarray*}
where we have used $(1- \omega) \tilde{\tau} \in {\cal S} \bigl( \mathbb{R}_+,H^{\infty,\gamma_1}(B) \bigr)$ and the standard Sobolev embedding $H^m \rightarrow L^{\infty}$ for $m > \frac{5}{2}$.
The estimate is valid for all $N \in \mathbb{N}_0$, and therefore the corresponding symbol of the kernel function (\ref{tcompl}) belongs to $S^{-\infty}(\mathbb{R}^3 \times \mathbb{R}^3)$. 
\end{proof}

Likewise, we can decompose the effective interaction potential (\ref{Vkl2})
\[
 V^{(2)}({\bf x}_1,{\bf x}_2) = \omega(|{\bf x}_1 - {\bf x}_2|) V^{(2)}({\bf x}_1,{\bf x}_2)
 + \bigl( 1- \omega(|{\bf x}_1 - {\bf x}_2|) \bigr) V^{(2)}({\bf x}_1,{\bf x}_2) ,
\]
which contributes to linear and nonlinear terms (\ref{BG}), (\ref{oprod}) and (\ref{oprod2}). According to our assumption, the second part belongs to 
the Schwartz class ${\cal S} (\mathbb{R}^3 \times \mathbb{R}^3)$ and therefore represents a kernel function with 
corresponding symbol in $S^{-\infty}(\mathbb{R}^3 \times \mathbb{R}^3)$. This can be seen from the following simple argument. 
Let us first express the effective interaction potential in the canonical variables (\ref{xzcoord}), i.e.,
\[
 \tilde{V}^{(2)}({\bf x}, {\bf z})\equiv V^{(2)} ({\bf x}_1, {\bf x}_2) .
\]
and obtain the estimate
\begin{eqnarray*}
\lefteqn{\sup_{{\bf x}, {\bf z}} \biggl\{ \langle {\bf z} \rangle^N \left| \partial_{{\bf z}}^\alpha \partial_{{\bf x}}^\beta
 \bigl( 1- \tilde{\omega}(|{\bf z}|) \bigr) \tilde{V}_{kl}^{(2)}({\bf x}, {\bf z}) \right| \biggr\} \lesssim} \\ 
 & & \sup_{\substack{{\bf x}_1, {\bf x}_2 \\ \beta_1 +\beta_2 = \beta}} \biggl\{ \langle {\bf x}_1 , {\bf x}_2 \rangle^N
 \left| \partial_{{\bf x}_2}^{\alpha + \beta_1} \partial_{{\bf x}_1}^{\beta_2}
 \bigl( 1- \tilde{\omega}(|{\bf x}_1 - {\bf x}_2|) \bigr) V_{kl}^{(2)} ({\bf x}_1, {\bf x}_2) \right| \biggr\} < \infty ,
\end{eqnarray*}
for any value of the multi indices $\alpha, \beta$ and $N \in \mathbb{N}_0$,
which shows, cf.~Proposition \ref{SteinPr}, that the corresponding symbol of the effective interaction potential (\ref{Vkl2}) belongs to $S^{-\infty}(\mathbb{R}^3 \times \mathbb{R}^3)$.

\subsection{Asymptotic expansions and classical pseudo-differential operators}
In order to control the asymptotic behaviour of nonlinear terms, it is necessary to stick to classical pseudo-differential operators
with symbol classes $S^p_{cl}$, $p \in \mathbb{Z}$. 
It is convenient to rearrange the asymptotic expansion (\ref{asymp1}) in order to make the transition to classical symbols more transparent.
Introducing new coordinates $s_1=t\sin r$, $s_2=t \cos r$ and performing a Taylor expansion of the edge function
\[
  v_{jk}(t,\theta_2,\phi_2) = \sum_{i=0}^{N-1} \tfrac{1}{j!} s_1^j \left. \partial_{s_1}^j v_{jk} \bigl( \sqrt{s_1^2+s_2^2},\theta_2,\phi_2 \bigr) \right|_{s_1=0} + s_1^N g_N(s_1,s_2,\theta_2,\phi_2) ,
\]
with $g_N$ smooth in $s_1,s_2$, we arrive at an equivalent asymptotic expansion of the form
\begin{equation}
 \tilde{\omega}(s_1) \sum_{j=0}^{N_\Theta} \sum_{k=0}^{m_j} s_1^j \log^k s_1 \, c_{jk}(\theta_1,\phi_1) \tilde{v}_{jk}(s_2,\theta_2,\phi_2) + \tilde{h}_{\Theta}(s_1,s_2,\theta_1,\phi_1,\theta_2,\phi_2)
\label{asymp2}
\end{equation}

In the following, we have to represent the pair-amplitude with respect to ${\bf x},{\bf z}$-variables. 
According to our assumption, $\tilde{v}_{jk} \in {\cal S}\bigl(\mathbb{R}_+,H^\infty_{\comp}(Y)\bigr)$ depends only on ${\bf z}_2$. Let
\[
 \tilde{v}_{jk}(s_2,\theta_2,\phi_2) =: \hat{v}_{jk}({\bf z}_2) = \hat{v}_{jk} \bigl( \sqrt{2} {\bf x} - \tfrac{1}{\sqrt{2}} {\bf z} \bigr)
\]
and perform another Taylor expansion at ${\bf z}=0$, i.e.,
\[
 \hat{v}_{jk} ({\bf z}_2) = \sum_{|\alpha| < n} \frac{(-1)^{|\alpha|}}{\alpha ! 2^{|\alpha|/2}} {\bf z}^\alpha \bigl(
 \partial^\alpha_{{\bf z}_2} \hat{v}_{jk} \bigr) \bigl( \sqrt{2} {\bf x} \bigr) \, + \sum_{|\alpha| = n} {\bf z}^\alpha g_{\alpha} ({\bf x},{\bf z}) ,
\]
with  $g_{\alpha}$ smooth in ${\bf x}$ and ${\bf z}$. With this, we obtain another equivalent asymptotic expansion of the form
\begin{equation}
 \tilde{\omega}(s_1) \sum_{j=0}^{N_\Theta} \sum_{k=0}^{m_j} s_1^j \log^k s_1 \, c_{jk}(\theta_1,\phi_1) \check{v}_{jk}({\bf x}) + \check{h}_{\Theta}({\bf x},{\bf z})
\label{asymp3}
\end{equation}

\begin{proposition}
Given a term of the asymptotic expansion (\ref{asymp3}) of the pair-amplitude 
\[
 \tau_{j0}({\bf x},{\bf z}) := \tilde{\omega}(s_1) s_1^j \, c_{j0}(\theta_1,\phi_1) \check{v}_{j0}({\bf x})
\]
Let us assume the orthogonality constraint
\begin{equation}
 \int_{0}^{2\pi} \int_{0}^{\pi} c_{j0}(\theta_1,\phi_1) \, Y_{lm}(\theta_1,\phi_1) \sin \theta_1 \, d\theta_1 d\phi_1 = 0
 \quad \mbox{for} \quad l>j.
\label{orthcon}
\end{equation}
with respect to spherical harmonics $Y_{lm}$ defined on the base $X$ of the cone.
Now let the corresponding symbol be given by
\[
 \sigma_{-3-j,0}({\bf x},\boldeta) := \int e^{-i {\bf z} \boldeta} \tau_{j0}({\bf x},{\bf z}) \, d{\bf z} .
\]
The symbol $\sigma_{-3,0}$ belongs $S_{cl}^{-\infty}(\mathbb{R}^3 \times \mathbb{R}^3)$ and $\sigma_{-3-j,0}$, $j > 0$, 
belongs to the symbol class $S_{cl}^{-3-j}(\mathbb{R}^3 \times \mathbb{R}^3)$ of classical pseudo-differential operators.
Furthermore, let $\boldeta$ be represented in spherical coordinates, i.e., $\sigma_{-3-j,0}({\bf x},\boldeta) \equiv \tilde{\sigma}_{-3-j,0}({\bf x},\eta,\Theta_1,\Phi_1)$, these symbols satisfy the orthogonality constraints
\begin{equation}
 \int_{0}^{2\pi} \int_{0}^{\pi} \tilde{\sigma}_{-3-j,0}({\bf x},\eta,\Theta_1,\Phi_1) \, Y_{lm}(\Theta_1,\Phi_1) \sin \Theta_1 \, d\Theta_1 d\Phi_1 = 0
 \quad \mbox{for} \quad j-l \ \mbox{even or} \ l>j.
\label{orthtau1}
\end{equation}
\label{aj0class}
\end{proposition}

\begin{proof}
Let us introduce spherical coordinates for the covariable $\boldeta$, denoted by $(\eta,\Theta,\Phi)$, 
and perform an expansion of the phase factor in terms of spherical harmonics, i.e.,
\[
 e^{-i {\bf z} \boldeta} = 4 \pi \sum_{l=0}^\infty \sum_{m=-l}^l (-i)^l j_l(\eta s_1) \bar{Y}_{lm}(\theta_1,\phi_1) Y_{lm}(\Theta,\Phi) ,
\]
where $j_l$ denote spherical Bessel functions. Inserting the expansion into the oscillatory integral for the symbol and taking into account the 
orthogonality constraint (\ref{orthcon}), we get
\begin{eqnarray*}
 \sigma_{-3-j,0}({\bf x},\boldeta) & = & \int e^{-i {\bf z} \boldeta} \tau_{j0}({\bf x},{\bf z}) \, d{\bf z} \\
 & = & \sum_{l=0}^j (-i)^l C_{l,j}(\Theta,\Phi) \check{v}_{j0}({\bf x}) \int_0^\infty j_l(\eta s_1) \tilde{\omega}(s_1) s_1^{j+2} \, ds_1 , 
\end{eqnarray*}
with 
\[
 C_{l,j}(\Theta,\Phi) := 4 \pi \sum_{m=-l}^l Y_{lm}(\Theta,\Phi) \int_{0}^{2\pi} \int_{0}^{\pi} c_{j0}(\theta_1,\phi_1) \, \bar{Y}_{lm}(\theta_1,\phi_1) \sin \theta_1 \, d\theta_1 d\phi_1 .
\]

First let us consider the trivial case $j=0$, with $j_0(\eta s_1) = \frac{\sin(\eta s_1)}{\eta s_1}$, we get
\[
 \sigma_{-3,0}({\bf x},\boldeta) = C_{0,0} \check{v}_{00}({\bf x}) \frac{1}{\eta} \int_0^\infty \sin (\eta s_1) \tilde{\omega}(s_1) s_1 ds_1 .
\]
For $f \in C^\infty_0(\mathbb{R})$ one gets in the limit $\eta \rightarrow \infty$, the asymptotic expansion, cf.~\cite{Copson}. 
\begin{equation}
 \int_0^\infty \sin(\eta s_1) f(s_1) \, ds_1 \sim \sum_{\substack{n \geq 0 \\ \even}} \frac{i^n}{\eta^{n+1}} f^{(n)}(0) .
\label{sinasymp}
\end{equation}
Obviously, one can easily extend $\tilde{\omega}(s_1) s_1$ to a function $f \in C^\infty_0(\mathbb{R})$ such that 
$f^{(n)}(0)=0$ for $n=0,2,\ldots$ and therefore $\sigma_{-3}$ belongs to $S^{-\infty}_{cl}$.

In order to treat the cases $j>0$, let us apply to spherical Bessel functions, with $l>0$, the recurrence relation, cf.~\cite{AS},
\begin{equation}
 j_l(\eta s_1) = - \frac{1}{\eta} \frac{d}{ds_1} j_{l-1}(\eta s_1) + \frac{l-1}{\eta s_1} j_{l-1}(\eta s_1) ,
\label{Bessrec}
\end{equation}
which yields
\begin{eqnarray*}
 \int_0^\infty j_l(\eta s_1) \tilde{\omega}(s_1) s_1^{j+2} \, ds_1
 & = & \frac{1}{\eta} \int_0^\infty \biggl( - \frac{d}{ds_1} j_{l-1}(\eta s_1) + \frac{l-1}{s_1} j_{l-1}(\eta s_1) \biggr)
 \tilde{\omega}(s_1) s_1^{j+2} \, ds_1 \\ 
 & \sim & \frac{1}{\eta} \int_0^\infty \tilde{\omega}(s_1) j_{l-1}(\eta s_1) 
 \bigl( l+j+1 \bigr) s_1^{j+1} \, ds_1 ,
\end{eqnarray*}
where we bear in mind that terms depending on derivatives of the cut-off function $\tilde{\omega}$ only contribute to $S^{-\infty}_{cl}$.
Here and in the following $\sim$ denotes equality modulo terms which belong to $S^{-\infty}_{cl}$.
After successive application of (\ref{Bessrec}), we get
\[
 \sigma_{-3-j,0}({\bf x},\boldeta) \sim \sum_{l=0}^j (-i)^l C_{l,j}(\Theta,\Phi) \check{v}_{j0}({\bf x}) \frac{a_{lj}}{\eta^{l+1}}
 \int_0^\infty \sin (\eta s_1) \tilde{\omega}(s_1) s_1^{j-l+1} ds_1 
\]
with
\begin{equation}
 a_{0j} = 1 \ \mbox{and} \  a_{lj} = \prod_{n=0}^{l-1} \bigl( j+l+1-2n \bigr) \ \mbox{for} \ l>0 .
\label{alj}
\end{equation}
For $j-l+1$ even, $\tilde{\omega}(s_1) s_1^{j-l+1}$ can be extended to a function $f \in C^\infty_0(\mathbb{R})$ such that
for even $n \neq j-l+1$ we get $f^{(n)}(0)=0$ and only $f^{(j-l+1)}(0) = (j-l+1)!$ does not vanish.
It then follows from (\ref{sinasymp}), that for $j>0$, we get 
\[
 \sigma_{-3-j,0}({\bf x},\boldeta) \sim \frac{1}{\eta^{j+3}} \check{v}_{j0}({\bf x}) \sum_{\substack{l=0 \\ \even}}^{j-1} (-1)^{\frac{j-l+1}{2}} (j-l+1)! \, a_{lj} \, C_{l,j}(\Theta,\Phi) 
 \quad \quad (j \ \mbox{odd}) ,
\]
and
\[
 \sigma_{-3-j,0}({\bf x},\boldeta) \sim \frac{1}{\eta^{j+3}} \check{v}_{j0}({\bf x}) \sum_{\substack{l=1 \\ \odd}}^{j-1} (-1)^{\frac{j-l+1}{2}} (j-l+1)! \, a_{lj} \, C_{l,j}(\Theta,\Phi) 
 \quad \quad (j \ \mbox{even}) ,
\]
respectively. The orthogonality constraints (\ref{orthtau1}) are an immediate consequence of these asymptotic relations.
\end{proof}

\begin{proposition}
Given a logarithmic term of the asymptotic expansion (\ref{asymp3})
\[
 \tau_{j1}({\bf x},{\bf z}) := \tilde{\omega}(s_1) s_1^j \log s_1 \, c_{j1}(\theta_1,\phi_1) \check{v}_{j1}({\bf x})
 \quad (j>0) ,
\]
where we assume the orthogonality constraint
\[
 \int_{0}^{2\pi} \int_{0}^{\pi} c_{j1}(\theta_1,\phi_1) \, Y_{lm}(\theta_1,\phi_1) \sin \theta_1 \, d\theta_1 d\phi_1 = 0
 \quad \mbox{for} \quad j-l \ \mbox{odd or} \ l>j.
\]
with respect to spherical harmonics $Y_{lm}$ defined on the base $X$ of the cone.
Now let the corresponding symbol be given by
\[
 \sigma_{-3-j,1}({\bf x},\boldeta) := \int e^{-i {\bf z} \boldeta} \tau_{j1}({\bf x},{\bf z}) \, d{\bf z} .
\]
The symbol $\sigma_{-3-j,1}$ belongs to the symbol class $S_{cl}^{-3-j}$ of classical pseudo-differential operators.
\label{aj1class}
\end{proposition}

\begin{proof}
We literally repeat the first steps of the proof of Prop.~\ref{aj0class}, The recurrence relation (\ref{Bessrec}) yields
\begin{eqnarray*}
 \int_0^\infty j_l(\eta s_1) \tilde{\omega}(s_1) s_1^{j+2} \log s_1 \, ds_1
 & = & \frac{1}{\eta} \int_0^\infty \biggl( - \frac{d}{ds_1} j_{l-1}(\eta s_1) + \frac{l-1}{s_1} j_{l-1}(\eta s_1) \biggr)
 \tilde{\omega}(s_1) s_1^{j+2} \log s_1 \, ds_1 \\ 
 & \sim & \frac{1}{\eta} \int_0^\infty \tilde{\omega}(s_1) j_{l-1}(\eta s_1) \biggl[ \bigl( l+j+1 \bigr) s_1^{j+1} \log s_1 + s_1^{j+1} \biggr] \, ds_1 \\
 & \sim & \frac{1}{\eta} \int_0^\infty \tilde{\omega}(s_1) j_{l-1}(\eta s_1) \bigl( l+j+1 \bigr) s_1^{j+1} \log s_1 \, ds_1 .
\end{eqnarray*}
In the last step we have used the fact that
according to our orthogonality constraint only terms with $j-l$ even have to be taken into account and therefore any term without
logarithm contributes to $S^{-\infty}_{cl}$. After successive application of (\ref{Bessrec}) and further partial integrations, we get
\begin{eqnarray*}
 \sigma_{-3-j,1}({\bf x},\boldeta) & \sim & \sum_{l=0}^j i^l C_{l,j1}(\Theta,\Phi) \check{v}_{j1}({\bf x}) \frac{a_{lj}}{\eta^{l+1}}
 \int_0^\infty \sin (\eta s_1) \tilde{\omega}(s_1) s_1^{j-l+1} \log s_1 \, ds_1 \\
 & \sim & \frac{i^j}{\eta^{j+2}} \biggl[ \sum_{l=0}^j C_{l,j1}(\Theta,\Phi) \check{v}_{j1}({\bf x}) a_{lj} (j-l+1)! \biggr] \int_0^\infty \cos (\eta s_1) \tilde{\omega}(s_1) \log s_1 \, ds_1 \\
 & \sim & -\frac{i^j}{\eta^{j+3}} \biggl[ \sum_{l=0}^j C_{l,j1}(\Theta,\Phi) \check{v}_{j1}({\bf x}) a_{lj} (j-l+1)! \biggr] \int_0^\infty \sin (\eta s_1) \tilde{\omega}(s_1) s_1^{-1} \, ds_1 \\
 & \sim & -\frac{i^j}{\eta^{j+3}} \biggl[ \sum_{l=0}^j C_{l,j1}(\Theta,\Phi) \check{v}_{j1}({\bf x}) a_{lj} (j-l+1)! \biggr] \frac{\pi}{2} ,
\end{eqnarray*}
with constant $a_{lj}$ given by (\ref{alj}).
\end{proof}

Once we have established the correspondence between pair-amplitudes and classical pseudo-differential operators, it remains to demonstrate
a similar correspondence for the effective interaction potential $V^{(2)}({\bf x}_1,{\bf x}_2)$ 
which is conveniently expressed in ${\bf x},{\bf z}$-coordinates, cf.~(\ref{xzcoord}), via $\tilde{V}^{(2)}({\bf x},{\bf z})$
Introducing canonical variables, let us perform 
a Taylor expansion of the orbital, i.e.,
\[
 \phi({\bf x}_2) = \phi({\bf x} - {\bf z})
 = \sum_{|\alpha| \leq n} \tfrac{(-1)^{|\alpha|}}{\alpha !} \mathbf{z}^\alpha \biggl( \partial^\alpha_{x_1} \phi \biggr) \bigl( \mathbf{x} \bigr)
 + R_n(\mathbf{x}, \mathbf{z}) ,
\]
which yields the asymptotic expansion 
\begin{eqnarray}
\nonumber
 \omega(|{\bf z}|) \tilde{V}^{(2)}({\bf x},{\bf z}) & = & \omega(|{\bf z}|) \frac{1}{|\mathbf{z}|} \sum_{|\alpha| \leq n} \tfrac{(-1)^{|\alpha|}}{\alpha !}
 \mathbf{z}^\alpha \phi'({\bf x}) \biggl( \partial^\alpha_{x_1} \phi \biggr) \bigl( \mathbf{x} \bigr) + \tilde{R}_n(\mathbf{x}, \mathbf{z}) \\ \label{Vasymp}
 & = & \omega(|{\bf z}|) \sum_{j=0}^n |\mathbf{z}|^{j-1} \sum_{\substack{l \leq j\\ j-l \, \even}} \sum_{m=-l}^l 
 Y_{lm}(\theta_1, \phi_1) v_{lm}(\mathbf{x}) + \tilde{R}_n(\mathbf{x}, \mathbf{z}) ,
\end{eqnarray}
for later purpose expressed in spherical coordinates,
where $v_{lm}$ belongs to the Schwartz class ${\cal S} (\mathbb{R}^3)$.

\begin{proposition}
The short range part of the effective interaction potential $\omega(|{\bf z}|) \tilde{V}^{(2)}({\bf x},{\bf z})$ represents the 
kernel function of a classical pseudo-differential operator with corresponding symbol in $S_{cl}^{-2}$. 
Given a term of the asymptotic expansion (\ref{Vasymp}) of the effective interaction potential
\[
 V_j(\mathbf{x}, \mathbf{z}) :=  \omega(s_1) s_1^{j-1} \sum_{\substack{l \leq j\\ j-l \, \even}} \sum_{m=-l}^l
 Y_{lm}(\theta_1, \phi_1) v_{lm}(\mathbf{x}) .
\]
The corresponding symbol is given by
\begin{eqnarray}
\nonumber
 \rho_{-2-j}({\bf x},\boldeta) & := & \int e^{-i {\bf z} \boldeta} V_j({\bf x},{\bf z}) \, d{\bf z} \\ \label{rhoVj}
 & \sim & \frac{1}{\eta^{j+2}} {\cal D}_j(\Theta,\Phi,{\bf x})
\end{eqnarray}
with
\[
 {\cal D}_j(\Theta,\Phi,{\bf x}) := 
 4 \pi i^j \sum_{\substack{l \leq j\\ j-l \, \even}} \sum_{m=-l}^l Y_{lm}(\Theta,\Phi)
 v_{lm}({\bf x}) a_{lj} (j-l)! .
\]
The symbol $\rho_{-2-j}$ belongs to the symbol class $S_{cl}^{-2-j}$ of classical pseudo-differential operators.
\label{vjclass}
\end{proposition}

\begin{proof}
Let us first establish the properties of the asymptotic terms.
The following proof is essentially a literal copy of the proof of Prop.~\ref{aj0class}.
Inserting the expansion of the phase factor in terms of spherical harmonics, we get 
\[
 \rho_{-2-j}({\bf x},\boldeta) 
 = 4 \pi \sum_{\substack{l \leq j\\ j-l \, \even}} i^l \left[ \sum_{m=-l}^l \bar{Y}_{lm}(\Theta,\Phi)
 (-1)^m v_{l-m}({\bf x}) \right] \int_0^\infty j_l(\eta s_1) \omega(s_1) s_1^{j+1} \, ds_1 ,
\]
where we assume the phase convention $\bar{Y}_{lm} = (-1)^m Y_{l-m}$.
Application of the recurrence relation (\ref{Bessrec}) yields
\[
 \int_0^\infty j_l(\eta s_1) \omega(s_1) s_1^{j+1} \, ds_1 
 \sim \frac{a_{lj}}{\eta^{l+1}} \int_0^\infty \sin (\eta s_1) \omega(s_1) s_1^{j-l} ds_1 ,
\]
with constant $a_{lj}$ given by (\ref{alj}).
For $j-l$ even, $\omega(s_1) s_1^{j-l}$ can be extended to a function $f \in C^\infty_0(\mathbb{R})$ such that
for even $n \neq j-l$ we get $f^{(n)}(0)=0$ and only $f^{(j-l)}(0) = (j-l)!$ does not vanish.
From (\ref{sinasymp}), we finally get
\[
 \rho_{-2-j}({\bf x},\boldeta) \sim
 \frac{4 \pi i^j}{\eta^{j+2}} \sum_{\substack{l \leq j\\ j-l \, \even}} \sum_{m=-l}^l \bar{Y}_{lm}(\Theta,\Phi)
 (-1)^m v_{l-m}({\bf x}) a_{lj} (j-l)! .
\]

It remains to show that the corresponding symbol of the remainder $\tilde{R}_n(\mathbf{x}, \mathbf{z})$ of the asymptotic expansion (\ref{Vasymp}) 
actually belongs to  the symbol class $S_{cl}^{-3-n}$. For this let us consider another asymptotic expansion with $\tilde{n} > n$ of the orbital $\phi({\bf x})$.
It can be shown, cf.~Prop.~4.3.2 \cite{BS}, that its remainder satisfies the estimate
\[
 \| R_{\tilde{n}}({\bf x}, \cdot) \|_{L^\infty(B)} \lesssim d^{\tilde{n}-\frac{1}{2}} \| \phi({\bf x}, \cdot ) \|_{H^{\tilde{n}+1}(B)} ,
\]
where $B$ represents a ball with radius $d$ and center at the origin. The corresponding remainder $\tilde{R}_{\tilde{n}}(\mathbf{x}, \mathbf{z})$ of the asymptotic expansion (\ref{Vasymp}) 
therefore satisfies the estimates
\[
 \left| \partial_{\bf x}^\beta \partial_{\bf z}^\alpha \tilde{R}_{\tilde{n}}({\bf x}, {\bf z}) \right| \lesssim 
 |z|^{-\frac{3}{2}+\tilde{n}-|\alpha|} 
\]
which according to Proposition (\ref{SteinPr}) corresponds to the fact that its symbol belongs to  the symbol class $S_{cl}^{-\frac{3}{2}-\tilde{n}}$. 
This is sufficient for our purposes, let us just take $\tilde{n} =n+2$ and observe that the symbol of $\tilde{R}_{\tilde{n}}(\mathbf{x}, \mathbf{z})-\tilde{R}_{n}(\mathbf{x}, \mathbf{z})$
belongs to the symbol class $S_{cl}^{-3-n}$. 
\end{proof}

\subsection{Kernel functions of classical pseudo-differential operators}
The iterative solution of the RPA-CC equation requires the solution of Eq.~\ref{Ataun} in each 
iteration step, where the RPA terms on the right hand side will be identified with kernel functions of classical 
pseudo-differential operators. Concerning the general properties of classical pseudo-differential operators we refer to \cite{ES97},
and to \cite{GVF01} for a discussion of asymptotic properties of their kernel functions. The asymptotic expansion of the kernel functions
corresponds to the decomposition of the symbols into homogeneous parts. In order to obtain a kernel function from a homogeneous symbol
it is necessary to apply an appropriate regularization technique at $\eta \rightarrow 0$, cf.~\cite{CK10,EK85,EK89}, which gives rise to  
logarithmic terms in t their asymptotic expansions. Let us briefly recall where the logarithmic terms in the kernel functions come from.  

\begin{proposition}
Given a homogeneous symbol $\sigma_{p-j}$ from the asymptotic expansion (\ref{asymbcl}) of a classical symbol $\sigma_p \in S^p_{cl}(\mathbb{R}^3,\mathbb{R}^3)$, with $p < -3$.
Let us assume that for $\eta$ greater some constant it has the form
\[
 \sigma_{p-j}({\bf x},\boldeta) =  \frac{1}{\eta^{p+j}} \sum_{l=0}^j \sum_{m=-l}^l w_{j,lm}({\bf x}) \, Y_{lm}(\Theta,\Phi) .
\]
Now let the corresponding kernel function be given by the oscillatory integral
\[
 k_{p-j}({\bf x},{\bf z}) := \int e^{i {\bf z} \boldeta} \sigma_{p-j}({\bf x},\boldeta) \, \dbar{\bf z} ,
\]
which, by assumption, is absolutely convergent. The singular part of the kernel function is then given by
\[
 k_{p-j}({\bf x},{\bf z}) \sim \tilde{\omega}(s_1) s_1^{j-p-3} \sum_{\substack{l=0\\(j-p-l \, \mbox{even})}}^{j-p-3} \sum_{m=-l}^l \tilde{w}_{j,lm}({\bf x}) \, Y_{lm}(\theta,\phi) ,
\]
\[
 k_{p-j}({\bf x},{\bf z}) \sim \tilde{\omega}(s_1) s_1^{j-p-3} \log s_1 \sum_{\substack{l=0\\(j-p-l \, \mbox{odd})}}^j \sum_{m=-l}^l \tilde{w}_{j,lm}({\bf x}) \, Y_{lm}(\theta,\phi) ,
\]
where $\sim$ denotes equality modulo terms which belong to $C^{\infty}(\mathbb{R}^3,\mathbb{R}^3)$.
\label{kernelasymp}
\end{proposition}

\begin{proof}
The proof is a simple consequence of Props.~(\ref{aj0class}) and (\ref{aj1class}), where we have to replace $j$ by $j-p-3$.
The homogeneous symbols considered in these Propositions
equal the present symbol modulo symbols in $S^{-\infty}_{cl}(\mathbb{R}^3,\mathbb{R}^3)$ which correspond to
$C^{\infty}(\mathbb{R}^3,\mathbb{R}^3)$ kernel functions. By taking multiplicative constants, it is therefore straightforeward to adjust $\tilde{w}_{j,lm}({\bf x})$ 
such that the calculations in the proofs of  Props.~(\ref{aj0class}) and (\ref{aj1class}) lead to our homogeneous symbol $\sigma_{p-j}({\bf x},\boldeta)$
modulo symbols in $S^{-\infty}_{cl}(\mathbb{R}^3,\mathbb{R}^3)$.
\end{proof}

\begin{proposition}
The kernel function $k_{p-N}$ corresponding to the remaining symbol $\sigma_{p-N}$ of the asymptotic expansion (\ref{asymbcl}) belongs modulo a smooth part to all wedge Sobolev spaces 
${\cal W}_{\loc}^\infty(Y, {\cal K}^{\infty, \tilde{\gamma}}(X^\wedge))$ with weight $\tilde{\gamma} < -\tfrac{3}{2}+N-p$.
If $k_{p-N}$ furthermore satisfies (\ref{xbound}) it belongs modulo a smooth part to the corresponding Schwartz spaces ${\cal S} \bigl( \mathbb{R}_+,H^{\infty,\tilde{\gamma}}(B) \bigr)$.
\label{kpN}
\label{Wremainder}
\end{proposition}
 
\begin{proof}
Before studying regularity in weighted wedge Sobolev spaces, it is necessary to adjust the kernel function such that it vanishes to certain order at ${\bf z} \rightarrow 0$. 
This can be achieved by subtracting a smooth kernel function, i.e. a polynomial times a cut-off function from it. In order to achieve a fairly optimal behaviour,
we make use of polynomial approximation in Sobolev spaces, cf.~\cite{BS}. On a ball $\tilde{B}$ of radius $d$ centered at the origin one gets the decomposition
\begin{equation}
 k_{p-N}({\bf x},{\bf z}) = Q_{p-N}({\bf x},{\bf z}) + R_{p-N}^m({\bf x},{\bf z}) ,
\end{equation}
with $Q_{p-N}$ a polynomial of degree less than $m$ in ${\bf z}$ with coefficients which are smooth functions in ${\bf x}$.
It can be shown, cf.~Prop.~4.3.2 \cite{BS}, that the remainder satisfies the following estimate
\[
 \| R_{p-N}^m({\bf x}, \cdot) \|_{L^\infty(\tilde{B})} \lesssim d^{m-\frac{3}{2}} \| k_{p-N}({\bf x}, \cdot ) \|_{H^m(\tilde{B})} .
\]
In order to estimate the Sobolev regularity of the kernel function let us consider the estimate
\begin{eqnarray*}
 \| k_{p-N}({\bf x}, \cdot) \|_{H^s(\tilde{B})} & = & \lesssim \int^{\mathbb{R}^3} \left| \langle \boldeta \rangle^s \sigma_{p-N}({\bf x},\boldeta) \right|^2 \, d \boldeta \\
 & \lesssim & \int^{\mathbb{R}^3} \langle \boldeta \rangle^{2s} \langle \boldeta \rangle^{2p-2N} \eta^2 d\eta \\
 & < & \infty \quad \mbox{for} \quad s < -\tfrac{1}{2} -p+N ,
\end{eqnarray*}
which means that $k_{p-N}({\bf x}, \cdot)$ belongs to $H^s(\tilde{B})$ for $s < -\tfrac{1}{2} -p+N$.
In the following, we will need the Sobolev regularity of the remainder and of its mixed partial derivatives. The latter also follows from the previous estimate
by a slight modification of the arguments,
which shows that $\partial_{\bf x}^\beta \partial_{\bf z}^\alpha R_{p-N}$ belongs to $H^s(\tilde{B})$ for $s < -\tfrac{1}{2} -|\alpha| -p+N$ 
and we can take $m = -1-|\alpha| -p+N$.

According to our previous discussion, let us consider in the following a modified kernel function
\[
 \tilde{k}_{p-N}({\bf x}, {\bf z}) := R_{p-N}^m({\bf x}, {\bf z}) ,
\] 
which satisfies the estimates
\begin{equation}
 \left| \partial_{\bf x}^\beta \partial_{\bf z}^\alpha \tilde{k}_{p-N}({\bf x}, {\bf z}) \right| \lesssim 
 \left\{ \begin{array}{ll} |z|^{-\frac{5}{2}-|\alpha| -p+N} & \mbox{for} \ |\alpha| \leq -3-p+N \\
 |z|^{-3-|\alpha| -p+N} & \mbox{for} \ |\alpha| > -3 -p+N \end{array} \right. .
\label{estpd}
\end{equation}
In order to show that it belongs to the wedge Sobolev space ${\cal W}_{\loc}^\infty(Y, {\cal K}^{\infty, \tilde{\gamma}}(X^\wedge))$
with weight $\tilde{\gamma}$, we have to consider the system of weighted local semi norms
\[
 \| \tilde{k}_{p-N} \|_{\alpha,\beta} :=
 \int_{X^\wedge} \int_{Y} r^{2-2\tilde{\gamma}} \left| \sigma(r) \phi(y) \partial^{\alpha}_{X^\wedge} \partial^{\beta}_{{\cal Y}} \tilde{k}_{p-N} \right|^2 \, dr dx dy .
\]
Partial derivatives in hyperspherical coordinates can be estimated by partial derivatives in Cartesian coordinates via the estimate
\begin{eqnarray*}
 \left| \partial^{\alpha}_{X^\wedge} \partial^{\beta}_{Y} \tilde{k}_{p-N} \right| & \lesssim &
 \sup_{\substack{\alpha=\alpha_1+\alpha_2 \\ \beta=\beta_1+\beta_2}} r^{|\alpha| +|\beta_2|} \left| \partial^{\alpha_1+\beta_1}_{\bf x} \partial^{\alpha_2+\beta_2}_{{\bf z}} k_{p-N} \right| \\
 & \lesssim & \sup_{\alpha=\alpha_1+\alpha_2} r^{|\alpha_1|-3-p+N} \\
 & \lesssim & r^{-3-p+N} ,
\end{eqnarray*}
where we used (\ref{estpd}) and
\[
 \partial_{X^\wedge} {\bf x} = {\cal O}(r), \quad \partial_{Y} {\bf x} = {\cal O}(1), \quad \partial_{X^\wedge} {\bf z} = {\cal O}(r), \quad \partial_{Y} {\bf z} = {\cal O}(r) . 
\]
Now we can easily estimate the wedge Sobolev norm
\begin{eqnarray*}
 \| \tilde{k}_{p-N} \|_{\alpha,\beta} & \lesssim & \int_{\supp \sigma} r^{2-2\tilde{\gamma}} r^{-6-2p+2N} dr \\
 & < & \infty \quad \mbox{for} \quad \tilde{\gamma} < -\tfrac{3}{2}+N-p .
\end{eqnarray*}

Let us finally show that it also belongs to the corresponding Schwartz spaces  ${\cal S} \bigl( \mathbb{R}_+,H^{\infty,\tilde{\gamma}}(B) \bigr)$ if it satisfies
(\ref{xbound}). Here we have to consider the system of weighted local semi norms
\[
 \| \tilde{k}_{p-N} \|_{\alpha,\beta,n} := \sup_{t>c>0}
 \int_{X^\wedge} \int_{\cal Y} r^{2-2\tilde{\gamma}} \left| \sigma(r) \phi(y) t^n \partial^{\alpha}_{X^\wedge} \partial^{\beta}_{{\cal Y}} \tilde{k}_{p-N} \right|^2 \, dr dx dy .
\]
The previous estimate of partial derivatives in hyperspherical coordinates can be modified according to
\begin{eqnarray*}
 \sup_{t>c>0} \left| t^n \partial^{\alpha}_{X^\wedge} \partial^{\beta}_{{\cal Y}} \tilde{k}_{p-N} \right| & \lesssim &
 \sup_{\substack{\alpha=\alpha_1+\alpha_2 \\ \beta=\beta_1+\beta_2}} r^{|\alpha| +|\beta_2|} \left| \bigl[ \max \{|{\bf x}|,|{\bf z}|\} \bigr]^{n+|\alpha|+|\beta|} \partial^{\alpha_1+\beta_1}_{\bf x} \partial^{\alpha_2+\beta_2}_{{\bf z}} \tilde{k}_{p-N} \right| \\
 & \lesssim & \sup_{\alpha=\alpha_1+\alpha_2} r^{|\alpha_1|-3-p+N} \\
 & \lesssim & r^{-3-p+N} ,
\end{eqnarray*}
from which we obtain as before
\begin{eqnarray*}
 \| \tilde{k}_{p-N} \|_{\alpha,\beta,n} < \infty \quad \mbox{for} \quad \tilde{\gamma} < -\tfrac{3}{2}+N-p .
\end{eqnarray*}
\end{proof}

\begin{corollary}
If a  kernel function $k_{p}$ which corresponds to a symbol in $S^p_{\cl}(\mathbb{R}^3 \times \mathbb{R}^3)$ satisfies (\ref{xbound}) it belongs modulo a smooth part to the Schwartz spaces ${\cal S} \bigl( \mathbb{R}_+,H^{\infty,\tilde{\gamma}}(B) \bigr)$ with $\tilde{\gamma} < -\tfrac{3}{2}-p$.
\label{Sclassical}
\end{corollary}

\section{Asymptotic properties of RPA type interactions}
\label{asympRPA}
Based on the results of the previous section for individual building blocks of the RPA interaction terms, we can now classify these 
terms either as kernel functions of classical pseudo-differential operators or in the framework of singular analysis. Since we are now
considering a real physical model, we reintroduce indices of occupied orbitals in our formulas. However, we still ignore spin degrees of freedom.

\begin{lemma}
Given a pair-amplitude $\tau_{kl}$ which belongs to the Schwartz space ${\cal S}\bigl( \mathbb{R}_+,H^{\infty,\gamma}_P(B) \bigr)$, cf.~(\ref{Schwartztype}). 
Let the corresponding symbol be given by
\[
 \sigma_{ik}({\bf x},\boldeta) := \int e^{-i {\bf z} \boldeta} \tau_{ik}({\bf x},{\bf z}) \, d{\bf z} .
\]
The symbol $\sigma_{ik}$ belongs to the symbol class $S_{cl}^{-4}(\mathbb{R}^3 \times \mathbb{R}^3)$ of classical pseudo-differential operators
and can be asymptotically represented by ``homogeneous'' symbols
\[
 \sigma_{ik}({\bf x},\boldeta) \sim \sum_{0 \leq j} \sigma_{ik,-4-j}({\bf x},\boldeta) ,
\]
which means that for $\lambda \geq 1$ and $\eta$ greater some constant, we have 
\[
 \sigma_{ik,-4-j}({\bf x}, \lambda \boldeta) = \lambda^{-4-j} \sigma_{ik,-4-j}({\bf x},\boldeta) .
\]
Furthermore, let $\boldeta$ be represented in spherical coordinates, i.e., $\sigma_{ik,-4-j}({\bf x},\boldeta) \equiv \tilde{\sigma}_{ik,-4-j}({\bf x},\eta,\Theta_1,\Phi_1)$, these symbols satisfy the orthogonality constraints
\begin{equation}
 \int_{0}^{2\pi} \int_{0}^{\pi} \tilde{\sigma}_{ik,-4-j}({\bf x},\eta,\Theta_1,\Phi_1) \, Y_{lm}(\Theta_1,\Phi_1) \sin \Theta_1 \, d\Theta_1 d\Phi_1 = 0
 \quad \mbox{for} \quad j-l \ \mbox{odd or} \ l>j.
\label{orthtau2}
\end{equation}
\label{lemmatau}
\end{lemma}

\begin{proof}
This is an immediate consequence of the asymptotic expansion (\ref{asymp2}) and Propositions~\ref{Schwartz} and \ref{aj0class}.
\end{proof}

\begin{lemma}
Let the corresponding symbols of an effective interaction potential $V^{(2)}_{kj}({\bf x}_1,{\bf x}_2) \equiv \tilde{V}^{(2)}_{kj}({\bf x},{\bf z})$ be given by
\[
 \rho_{kj}({\bf x},\boldeta) := \int e^{-i {\bf z} \boldeta} \tilde{V}^{(2)}_{kj}({\bf x},{\bf z}) \, d{\bf z} .
\]
The symbol $\rho_{kj}$ belongs to the symbol class $S_{cl}^{-2}(\mathbb{R}^3 \times \mathbb{R}^3)$ of classical pseudo-differential operators
and can be asymptotically represented by ``homogeneous'' symbols
\[
 \rho_{kj}({\bf x},\boldeta) \sim \sum_{n \geq 0} \rho_{kj,-2-n}({\bf x},\boldeta) ,
\]
which means that for $\lambda \geq 1$ and $\eta$ greater some constant, we have 
\[
 \rho_{kj,-2-n}({\bf x}, \lambda \boldeta) = \lambda^{-2-j} \sigma_{kj,-2-n}({\bf x},\boldeta) .
\]
Furthermore, let $\boldeta$ be represented in spherical coordinates, i.e., $V^{(2)}_{kj}({\bf x},\boldeta) \equiv \tilde{V}^{(2)}_{kl,j}({\bf x},\eta,\Theta_1,\Phi_1)$, the asymptotic symbols satisfy 
the orthogonality constraints
\begin{equation}
 \int_{0}^{2\pi} \int_{0}^{\pi} \tilde{\rho}_{kl,-2-n}({\bf x},\eta,\Theta_1,\Phi_1) \, Y_{lm}(\Theta_1,\Phi_1) \sin \Theta_1 \, d\Theta_1 d\Phi_1 = 0
 \quad \mbox{for} \quad n-l \ \mbox{odd or} \ l>n.
\label{orthtau3}
\end{equation}
\label{lemmaV}
\end{lemma}

\begin{proof}
This is an immediate consequence of the asymptotic expansion (\ref{Vasymp}), Proposition \ref{vjclass} and the discussion following
Proposition \ref{Schwartz}.
\end{proof}

With this, the RPA term (\ref{RPA2}) becomes
\begin{eqnarray*}
 \int \tau_{ik}({\bf x}_1, {\bf x}_3) V_{kj}^{(2)}({\bf x}_3, {\bf x}_2) \, d{\bf x}_3 
 & = & \int \left( \int e^{i ({\bf x}_1 - {\bf x}_3) \boldeta} \sigma_{ik}({\bf x}_1,\boldeta) \, \dbar \boldeta \right)
 \left( \int e^{i ({\bf x}_3 - {\bf x}_2) \tilde{\boldeta}} \rho_{kj}({\bf x}_3,\tilde{\boldeta}) \, \dbar \tilde{\boldeta} \right) d {\bf x}_3 \\
 & = & \int e^{i ({\bf x}_1 - {\bf x}_2) \boldeta} \sigma_{ik} \circ \rho_{kj}({\bf x}_1,\boldeta) \, \dbar \boldeta ,
\end{eqnarray*}
where the composite symbol 
\[
  \sigma_{ik} \circ \rho_{kj}({\bf x}_1,\boldeta) := \iint e^{i ({\bf x}_1 - {\bf x}_3)(\tilde{\boldeta} - \boldeta)} \sigma_{ik}({\bf x}_1,\tilde{\boldeta}) \rho_{kj}({\bf x}_3,\boldeta) \, \dbar \tilde{\boldeta} \, d {\bf x}_3
\]
can be represented by the asymptotic Leibniz product
\begin{equation}
 \sigma_{ik} \circ \rho_{kj}({\bf x}_1,\boldeta) \sim
 \sum_{\alpha} \frac{1}{(2\pi i)^{|\alpha|} \alpha !} \bigl( \partial_{\boldeta}^{\alpha} \sigma_{ik} \bigr) ({\bf x}_1,\boldeta) \bigl( \partial_{{\bf x_1}}^{\alpha} \rho_{kj} \bigr) ({\bf x}_1,\boldeta) ,
\label{Leibniz}
\end{equation}
which means that the difference
\[
 \sigma_{ik} \circ \rho_{kj}({\bf x}_1,\boldeta) - 
 \sum_{|\alpha|>N} \frac{1}{(2\pi i)^{|\alpha|} \alpha !} \bigl( \partial_{\boldeta}^{\alpha} \sigma_{ik} \bigr) ({\bf x}_1,\boldeta) \bigl( \partial_{{\bf x_1}}^{\alpha} \rho_{kj} \bigr) ({\bf x}_1,\boldeta) ,
\]
belongs to the Symbol class $S_{cl}^{-6-N}(\mathbb{R}^3 \times \mathbb{R}^3)$. 
Finally, let us consider the RPA term (\ref{RPA3}), which becomes
\begin{eqnarray*}
 \lefteqn{\iint \tau_{ik}({\bf x}_1, {\bf x}_3) V_{[k,l]}^{(2)}({\bf x}_3, {\bf x}_4) \tau_{lj}({\bf x}_4, {\bf x}_2) \, d{\bf x}_3 d{\bf x}_4} \hspace{1cm} \\
 & & = \int \left( \int e^{i ({\bf x}_1 - {\bf x}_4) \boldeta} \sigma_{ik} \circ \rho_{[k,l]}({\bf x}_1,\boldeta) \, \dbar \boldeta \right)
 \left( \int e^{i ({\bf x}_4 - {\bf x}_2) \tilde{\boldeta}} \sigma_{lj}({\bf x}_4,\tilde{\boldeta}) \, \dbar \tilde{\boldeta} \right) d {\bf x}_4 \\
 & & = \int e^{i ({\bf x}_1 - {\bf x}_2) \boldeta} \sigma_{ik} \circ \rho_{[k,l]} \circ \sigma_{lj}({\bf x}_1,\boldeta) \, \dbar \boldeta ,
\end{eqnarray*}
with $\rho_{[k,l]} := \rho_{kl} - \rho_{kl}$, where the associativity of the Leibniz product has been used in the last step.

To simplify our notation let us define
\begin{equation}
 \sigma^{\rpa}_{ij,-6} := \sum_k \sigma_{ik} \circ \rho_{kj}, \quad \quad \sigma^{\rpa}_{ij,-10} := \sum_{k,l} \sigma_{ik} \circ \rho_{[k,l]} \circ \sigma_{lj} ,
\label{sigmaRPAij}
\end{equation}
both symbols represent classical pseudo-differential operators with ``homogeneous'' symbols which satisfy certain orthogonality constraints stated in the following lemma. 

\begin{lemma}
The composite symbols $\sigma^{\rpa}_{ij,-6}$ and $\sigma^{\rpa}_{ij,-10}$ belong to $S_{cl}^{-6}$ and $S_{cl}^{-10}$, respectively. They have asymptotic expansions 
\[
 \sigma^{\rpa}_{ij,p}({\bf x}_1,\boldeta) \sim \sum_{0 \leq n} \sigma^{\rpa}_{ij,p-n}({\bf x}_1,\boldeta), \quad \quad p=-6,-10
\]
with ``homogeneous'' symbols which satisfy 
\[
 \sigma^{\rpa}_{ij,p-n}({\bf x}_1, \lambda \boldeta) = \lambda^{p-n} \sigma^{\rpa}_{ij,p-n}({\bf x}_1,\boldeta), \quad \quad p=-6,-10
\]
for $\lambda \geq 1$ and $\eta$ greater some constant.
Furthermore, let $\boldeta$ be represented in spherical coordinates, i.e., 
$\tilde{\sigma}^{\rpa}_{p-n}({\bf x}_1,\eta,\Theta_1,\Phi_1) \equiv \sigma^{\rpa}_{p-n}({\bf x}_1,\boldeta)$, $p=-6,-10$, these symbols satisfy the orthogonality constraints
\begin{equation}
 \int_{0}^{2\pi} \int_{0}^{\pi} \tilde{\sigma}^{\rpa}_{ij,p-n}({\bf x},\eta,\Theta_1,\Phi_1) \, Y_{lm}(\Theta_1,\Phi_1) \sin \Theta_1 \, d\Theta_1 d\Phi_1 = 0
 \quad \mbox{for} \quad n-l \ \mbox{odd or} \ l>n.
\label{orthtau4}
\end{equation}
\label{lemmacompsymb}
\end{lemma}

\begin{proof}
The proof is given for the symbol $ \sigma^{\rpa}_{ij,-6}$ and is completely analogous for $ \sigma^{\rpa}_{ij,-12}$.
According to our orthogonality constraints (\ref{orthtau2}), (\ref{orthtau3}), we get the asymptotic decompositions
\[
 \sigma_{ik,-4-j}({\bf x}, \boldeta) \sim \sum_{\substack{l \leq j\\ j-l \, \even}} \sum_{m=-l}^l
 \frac{1}{\eta^{j+l+4}} v_{j,lm}({\bf x}) Z_{lm}(\boldeta) ,
\]
\[
 \rho_{ik,-2-j}({\bf x}, \boldeta) \sim \sum_{\substack{l \leq j\\ j-l \, \even}} \sum_{m=-l}^l
 \frac{1}{\eta^{j+l+2}} w_{j,lm}({\bf x}) Z_{lm}(\boldeta) ,
\]
where $Z_{lm}$ denotes the homogeneous polynomial associated to a spherical harmonic function, i.e., $Z_{lm}(\boldeta) \equiv \eta^l Y_{lm}(\Theta_1,\Phi_1)$.
To get a better understanding of the asymptotic behaviour of the Leibniz product, let us consider the effect of $\boldeta$ derivatives on the ``homogeneous'' symbols 
$\sigma_{ik,-4-j}({\bf x}, \boldeta)$ a little closer. 
Taking partial derivatives $\partial_a$, $a=1,2,3$, one gets
\[
 \partial_{\eta_a} \sigma_{ik,-4-j}({\bf x}, \boldeta) \sim \sum_{\substack{l \leq j\\ j-l \, \even}} \sum_{m=-l}^l 
 \biggl( -\frac{j+l+4}{\eta^{j+l+6}} v_{j,lm}({\bf x}) \eta_a Z_{lm}(\boldeta) + \frac{1}{\eta^{j+l+4}} v_{j,lm}({\bf x}) \partial_a Z_{lm}(\boldeta) \biggr) ,
\]
and by taking into account the following decompositions
\[
 \eta_a Z_{lm}(\boldeta) = \sum_{\substack{l' \leq l+1\\ l+1-l' \, \even}} \sum_{m'=-l'}^{l'} c_{l'm'} \eta^{l+1-l'} Z_{l'm'}(\boldeta) ,
\]
\[
 \partial_a Z_{lm}(\boldeta) = \sum_{\substack{l' \leq l-1\\ l-1-l' \, \even}} \sum_{m'=-l'}^{l'} \tilde{c}_{l'm'} \eta^{l-1-l'} Z_{l'm'}(\boldeta) ,
\]
it can be written as
\begin{eqnarray*}
 \partial_{\eta_a} \sigma_{ik,-4-j}({\bf x}, \boldeta) & \sim & \sum_{\substack{l \leq j+1\\ j+1-l \, \even}} \sum_{m=-l}^l 
 \frac{1}{\eta^{j+l+5}} \tilde{v}_{j,lm}({\bf x}) Z_{lm}(\boldeta) \\
 & \sim & \frac{1}{\eta^{j+5}} \sum_{\substack{l \leq j+1\\ j+1-l \, \even}} \sum_{m=-l}^l \tilde{v}_{j,lm}({\bf x}) Y_{lm}(\Theta_1,\Phi_1) .
\end{eqnarray*}
Therefore, taking a partial derivative decreases the degree of homogeneity by one but preserves the orthogonality constraints.
By induction, we get
\[
 \partial^\alpha_{\eta} \sigma_{ik,-4-j}({\bf x}, \boldeta) \sim \frac{1}{\eta^{j+4+|\alpha|}} \sum_{\substack{l \leq j+|\alpha|\\ j+|\alpha|-l \, \even}} \sum_{m=-l}^l 
 \tilde{v}_{j,lm}({\bf x}) Y_{lm}(\Theta_1,\Phi_1) .
\]
Let us consider the product
\begin{eqnarray*}
 \lefteqn{\partial^\alpha_{\eta} \sigma_{in,-4-j}({\bf x}, \boldeta) \partial^\alpha_{\bf x} \sigma_{nk,-2-j'}({\bf x}, \boldeta)} \\
 & \sim & \frac{1}{\eta^{j+j'+6+|\alpha|}} \sum_{\substack{l \leq j+|\alpha|\\ j+|\alpha|-l \, \even}} 
 \sum_{\substack{l' \leq j'\\ j'-l' \, \even}}
 \sum_{m=-l}^l \sum_{m'=-l'}^{l'} \tilde{v}_{j,lm}({\bf x}) \partial^\alpha_{\bf x} \tilde{w}_{j',l'm'}({\bf x})  Y_{lm}(\Theta_1,\Phi_1)  Y_{l'm'}(\Theta_1,\Phi_1) \\
 & \sim & \frac{1}{\eta^{j+j'+6+|\alpha|}} \sum_{\substack{L \leq j+j'+|\alpha|\\ j+j'+|\alpha|-L \, \even}}
 \sum_{M=-L}^L \tilde{u}_{jj'\alpha,LM}({\bf x}) Y_{LM}(\Theta_1,\Phi_1) ,
\end{eqnarray*}
where we used the product formula
\[
 Y_{lm}(\Theta_1,\Phi_1)  Y_{l'm'}(\Theta_1,\Phi_1) = \sum_{\substack{L \leq l+l'\\ l+l'-L \, \even}} 
 \sum_{M=-L}^L C_{lm,l'm',LM}Y_{LM}(\Theta_1,\Phi_1)
\]
together with the orthogonality constraints (\ref{orthtau2}) where all possible combinations are listed in Table \ref{table1}.
The orthogonality constraints (\ref{orthtau4}) are an immediate consequence of the Leibniz product formula (\ref{Leibniz}).
\end{proof}

\begin{table}[t]
\label{table1}
\begin{center}
\caption{Orthogonality constraints for spherical harmonics of the Leibniz product (\ref{Leibniz}).} 
\vspace{0.5cm}
\begin{tabular}{|c|c|c|c|c|c|}
\hline \small
 $j+|\alpha|$ &  $j'$ &  $l$  &  $l'$ & $j+j'+|\alpha|$ & $L$ \tabularnewline \hline
     odd      &  odd  & odd   & odd   & even            & even \tabularnewline  
     odd      &  even & odd   & even  & odd             & odd  \tabularnewline  
     even     &  odd  & even  & odd   & odd             & odd  \tabularnewline  
     even     &  even & even  & even  & even            & even \tabularnewline  
\hline
\end{tabular}
\end{center}
\end{table}

\section{Application of singular analysis to CC theory}
In Section \ref{asympRPA} we have studied the asymptotic type of the right hand side of Eq.~(\ref{Ataun}) provided $\tau_n$ belongs to a certain generic asymptotic type which has been specified before, cf.~Section \ref{kernel-wedge}. 
What remains is the actual solution step which can be studied via an asymptotic parametrix and
corresponding Green operator according to the general scheme outlined in Section \ref{iterationscheme}.  
Concerning a general presentation of the underlying theory of pseudo-differential operators on manifolds with singularities, we refer to the monographs \cite{ES97,HS08,Schulze98}.

In the following, we want to consider an asymptotic parametrix for shifted edge degenerate Hamiltonian operators 
\begin{equation}
A_{\edge}:=H_{\edge} -\lambda , \quad \lambda \in \mathbb{R}
\label{Aedge}
\end{equation}
representing a (non) interacting electron pair. 
In ordinary Cartesian coordinates such Hamiltonians are of the form
\[
 H_{\edge} = - \tfrac{1}{2} \bigl( \Delta_1 + \Delta_2 \bigr) +V(\boldsymbol{x}_1,\boldsymbol{x}_2 )
\]  
where the potential term $V$ includes one and possibly two particle
interactions to be specified below. Expressed in our hyperspherical
coordinates the  Hamiltonian becomes 
\begin{eqnarray}
 H_{\edge} & = & r^{-2} \Big[ -\frac{1}{2 t^2} (- r
{\partial_{r}} )^2
 - \frac{h(r)}{2t^2} (- r {\partial_{r}})
 -\frac{1}{2} (r {\partial_t} )^2
 - \frac{5 r}{2t} (r {\partial_t} )  \nonumber\\
 & & - \frac{1}{2t^2 \cos^2 r} (r {\partial_{\theta_2}})^2
 - \frac{r \mbox{ctan} \, \theta_2}{2t^2 \cos^2 r} (r {\partial_
{\theta_2}} )
- \frac{1}{2t^2 \sin^2 \theta_2 \cos^2 r} (r {\partial_
{\phi_2}} )^2 \nonumber\\
 & & - \frac{r^2}{2t^2 \sin^2 r} \Delta_{X_1} + \frac{r}{t} v_{\edge} \Big] \label{H.edge}
\end{eqnarray}
with
\[
h(r) := 1+2r \, \mbox{tan} \, r-2r \, \mbox{ctan}\, r ,
\]
which means that the hyperradius $t$ is actually treated as yet another edge variable.
It should be mentioned, that the potential part $v_{\edge}$ is smooth with respect to $r$ up to $r=0$. The latter
assumption is crucial for the singular pseudo-differential calculus
to be applicable. It has been shown in Ref.~\cite{FFHS16}
that this is actually the case for common Coulomb potentials.

\subsection{Asymptotic parametrices for edge degenerate Hamiltonians} 
In Ref.~\cite{FFHS16}, we have derived an asymptotic parametrix
for the Hamiltonian (\ref{H.edge}) modulo Green operators in $L^{0}_G(M,\boldsymbol{g})$.
This actually represents the penultimate step in the asymptotic parametrix construction discussed in Ref.~\cite{FHS16}, cf.~Corollary 2.24 and Theorem 2.26 therein. For our purposes
it is sufficient to stop at this point because it already provides us with the desired insight into the asymptotic behaviour of iterated pair-amplitudes near coalescence points of electrons.
Furthermore we want to mention that the construction of the parametrix involves a regularization step, cf.~Appendix A in Ref.~\cite{FFHS16} for further details. This is justified because we apply the parametrix and Green operator to functions which belong to ${\cal W}^\infty_{\mathrm{loc}} \bigl(Y,{\cal K}^{\infty,\gamma}\big((S^2)^{\wedge}\big) \bigr)$. For the right hand side of (\ref{Ataun}) this follows from our discussion in Section \ref{asympRPA} and according to standard regularity theory, cf.~\cite{ES97,Schulze98}, this also follows for the iterated pair-amplitude
$\tau_{n+1}$.

A parametrix $P$ of a shifted edge degenerate Hamiltonian operator (\ref{Aedge}) belongs to a class of singular pseudo-differential
operators which can be written in the general form
\begin{equation}
 P = \sum_i \sigma' \varphi_i \Op_y (p) \varphi'_i \tilde{\sigma}'
 +(1-\sigma')P_{\interior}(1-\hat{\sigma}') ,
\label{A}
\end{equation}
with $p \in R^{-2}(Y \times \mathbb{R}^3, {\bf g})$ and
given cut-off functions $\hat{\sigma}' \prec \sigma' \prec \tilde{\sigma}'$. 
In order to calculate the parametrix it is convenient to make the following ansatz
for the parameter dependent Mellin pseudo-differential operator
\begin{equation}
 p(y,\eta) = \omega'_{1,\eta} r^2 p_M(y,\eta) \omega'_{0,\eta}
 + (1-\omega'_{1,\eta}) r^2 p_\psi(y,\eta) (1-\omega'_{2,\eta}) ,
\label{p-mellin}
\end{equation}
with cut-off functions $\omega'_2, \omega'_1, \omega'_0$ satisfying $\omega'_2 \prec \omega'_1 \prec \omega'_0$ 
where we write $\omega_\eta(r) := \omega ([\eta]r)$; here $[\eta]$ is any fixed strictly positive function in 
$C^\infty(\mathbb{R}^3)$ such that $[\eta] = |\eta|$ for $\eta \geq \epsilon$ for some $\epsilon >0$. For the parameter dependent Mellin part $p_M$ let us assume an asymptotic Mellin expansion
\[
 p_M(y,\eta) := \opm_M^{\gamma-3}(a_0^{(-1)})(y,\eta) + r   \opm_M^{\gamma-3}(a_1^{(-1)})(y,\eta) +r^2 \opm_M^{\gamma-3}(a_2^{(-1)})(y,\eta) +\dots ,
\]
where the Mellin pseudo-differential operators are of the form
\[
\opm_M^{\gamma-3} (a_i^{(-1)})(y,\eta)u(r)=\int_\mathbb{R}\int_0^\infty (r/r')^{-(7/2-\gamma +i\rho )}a_i^{(-1)}(r,r',7/2-\gamma +i\rho ,y,\eta)u(r')dr'/r'\dbar \rho , 
\]
$\dbar \rho=(2\pi )^{-1}d\rho,$ with Mellin amplitude function $a_i^{(-1)}(r,r',z,y,\eta)$, $i=0,1,2,\ldots,$ taking values in $L_{\textup{cl}}^{-2} (X;\Gamma _{1/2-\gamma } \times \mathbb{R}^6)$. 
This expression has to be interpreted as a Mellin oscillatory integral and representing a family of operators
\[
C_0^\infty (\mathbb{R}_+,C^\infty (X))\rightarrow C^\infty (\mathbb{R}_+,C^\infty (X)).
\]
in $L_{\textup{cl}}^{-2} (X^\wedge;\mathbb{R}^6).$
According to Remark 3 in \cite{FFHS16}, we ignore the second term of (\ref{p-mellin}) in the following considerations.
By a slight modification of the standard notation we incorporate into $p_M$ contributions from 
$R^{-2}_{M+G}(Y \times \mathbb{R}^3, {\bf g})$ as well.

The operator valued meromorphic symbols of the parametrix have asymptotic expansions 
\begin{align*}
 a_0^{(-1)} & \sim -2t^2 \biggl( \frac{1}{h_0} + \frac{r^2C_0}{h_0 \bigl( h_0-2(2w-7) \bigr)} - \frac{r P_{1,1}}{h_0 \bigl( h_0-2(2w-7) \bigr)}
 - r^2 \frac{2}{h_0 \bigl( h_0-2(2w-7) \bigr)} + \cdots \biggr) \\
 a_1^{(-1)} & \sim -2t^2 \biggl( \frac{2tZ_1}{h_0 \bigl( h_0 -(2w-6) \bigr)} + \frac{irC_1}{h_0 \bigl( h_0 -2(2w-7) \bigr)} - r \frac{10}{h_0 \bigl( h_0 -2(2w-7) \bigr)} 
 + \cdots \biggr) \\
 a_2^{(-1)} & \sim -2t^2 \biggl( \frac{(2tZ_1)^2} {\bigl( h_0 -2(2w-7) \bigr) \bigl( h_0 -(2w-6) \bigr) h_0} - \frac{1}{3} \frac{8(w-2) + \Delta_{S^2} -6tZ_2}{\bigl( h_0 -2(2w-7) \bigr) h_0} + \cdots \biggr)  
\end{align*}
with operator valued holomorphic symbol
\begin{equation}
 h_0 = (w-2)^2 - (w-2) +\Delta_{S^2} 
\label{h0} 
\end{equation}
and $C_0 := t^2 \tau^2+\Theta_2^2+\frac{\Phi_2^2}{\sin^2 \theta_2}$, $C_1 := -5t\tau -\cot \theta_2 \Theta_2$ and $P_{1,1} := 4it(r\tau)$, cf.~\cite{FFHS16}. The potential $V$ enters into the asymptotic parametrix
via the parameters $Z_1$, $Z_2$, which can be derived from its representation in hyperspherical coordinates 
\[
 V = \frac{1}{t r} v_{\edge}(r,\theta_1,\phi_1,\theta_2,\phi_2) 
\]
where $v_{\edge}$ is smooth with respect to $r$ up to $r=0$. 
With this, the coefficients are given by
\[
 Z_1 := v_{\edge}(0,\theta_1,\phi_1,\theta_2,\phi_2), \quad \quad Z_2 := -tE + \partial_r v_{\edge}(0,\theta_1,\phi_1,\theta_2,\phi_2) .
\]

Within the present work, we consider a non interacting Hamiltonian, cf.~(\ref{MP1}), i.e.
\[
 A_{\edge} := \mathfrak{f}_1 + \mathfrak{f}_2 - \epsilon_i -\epsilon_j 
\]
and assume a potential of the form
\[
 V({\bf x}_1,{\bf x}_2) = \mathfrak{v}_1({\bf x}_1) + \mathfrak{v}_2({\bf x}_2) :=
 \mathfrak{v}_C({\bf x}_1) + \mathfrak{v}_H({\bf x}_1) + \mathfrak{v}_x({\bf x}_1) + \mathfrak{v}_C({\bf x}_2) + \mathfrak{v}_H({\bf x}_2) + \mathfrak{v}_x({\bf x}_2) ,
\]
which consists of a Coulomb, Hartree and local exchange part. In this case $Z_1$ becomes zero.

\subsection{Asymptotic properties of iterated pair-amplitudes}
\label{iterpair}
Let us first consider the asymptotic behaviour of first-order M{\o}ller-Plesset pair-amplitudes near the electron-electron cusp.
Applying the asymptotic parametrix to the left of Eq.~(\ref{MP1}) yields
\begin{equation}
 \tau_{ij} = - G_l^{(ij)} \tau_{ij} - P^{(ij)} \mathfrak{Q} \frac{1}{|{\bf x}_1 - {\bf x}_2|} \Psi^{(1)}_{ij} .
\label{PMP1}
\end{equation}
The asymptotic expansion of the Green operator has been given in Ref.~\cite{FFHS16}. For the sake of the reader, we recapitulate the main result.

\begin{theorem}\label{FFHS16theorem}
The Green operator $g\in R_G^0 (Y\times\mathbb{R}^3,{\boldsymbol g})$, for weight $\frac{1}{2} < \gamma < \frac{3}{2}$, has a leading order asymptotic expansion of the form
\begin{eqnarray}
 g\hat{u}(y,\eta)&=&\sigma^\prime 2t^2\left[\left(1+rtZ_1+r^2\left(-2+\tfrac{1}{3}(tZ_1)^2+\tfrac{1}{3}tZ_2\right)\right){\cal P}_0{\cal Q}_{0,1}(\hat{u})(y,\eta)\right.\nonumber\\
 &&+\tfrac{1}{6}r^2 {\cal P}_0{\cal Q}_{0,2}(\hat{u})(y,\eta)+\left(\tfrac{1}{3}r+\tfrac{1}{6}tZ_1 r^2\right){\cal P}_1{\cal Q}_{1,1}(\hat{u})(y,\eta)\nonumber\\
 &&+\left.\tfrac{1}{5}r^2{\cal P}_2{\cal Q}_{2,1}(\hat{u})(y,\eta)-\tfrac{1}{30}r^2 {\cal P}_2{\cal Q}_{2,2}({\hat u})(y,\eta)\right]+{\cal O}(r^3)\label{Gasymp}
\end{eqnarray}
where $\varphi u\in {\cal W}^\infty_{\mathrm{comp}} \bigl(Y,{\cal K}^{\infty,\gamma}\big((S^2)^{\wedge}\big) \bigr)$ and $\hat{u}(r,\phi_1,\theta_1,\eta):=F_{y\rightarrow\eta}\phi_i u(r,\phi_1,\theta_1,y).$ Here, $P_l, l=0,1,2,\ldots,$
denote projection operators on subspaces which belong to eigenvalues $-l(l+1)$ of the Laplace-Beltrami operator on $S^2.$
\end{theorem}

For explicit expressions of the terms ${\cal Q}_{n,m}(\hat{u})(y,\eta)$
depending on edge variables and covariables, we refer to Ref.~\cite{FFHS16}.
Let us first consider the case of relative angular momentum $l=0$, where we have the asymptotic expansion
\begin{eqnarray*}
 {\cal P}_0 G_l^{(ij)} \tau_{ij}(r,y) & = & \int e^{iy \eta} {\cal P}_0 g_{ij}(y,\eta) \hat{\tau}_{ij}(\eta) \, \dbar \eta \\
 & = & \sigma' 2t^2 \int e^{iy \eta} {\cal P}_0 {\cal Q}_{0,1}(\hat{\tau}_{ij})(y,\eta) \, \dbar \eta + {\cal O}(r^2) ,
\end{eqnarray*}
with
\[
 \int e^{iy \eta} {\cal Q}_{0,1}(\hat{\tau}_{ij}) \, \dbar \eta
 = \int e^{iy \eta} M \bigl( \opm_M^{\gamma-1}(h_1^{(0)}) \tilde{\sigma} \hat{\tau}_{ij} \bigr) (0) \, \dbar \eta 
 + \int e^{iy \eta} M \bigl( \tilde{\sigma}' \sigma \opm_M^{\gamma-1}(a_{\om}) \tilde{\sigma} \hat{\tau}_{ij} \bigr) (0) \, \dbar \eta ,
\]
depending on the Mellin symbols $h_1^{(0)}(\omega) := \tfrac{1}{2t^2} ( w^2-w)$ and $a_{\om}(\omega)$, the latter denotes the symbol of the Hamiltonian in first-order M{\o}ller-Plesset perturbation theory.
We note that
\begin{eqnarray*}
 \int e^{iy \eta} M \bigl( \tilde{\sigma}' \sigma \opm_M^{\gamma-1}(a_{\om}) \tilde{\sigma} \hat{\tau}_{ij} \bigr) (0) \, \dbar \eta
 & = & M \bigl( \tilde{\sigma}' \sigma r^2(-\tfrac{1}{2} \Delta_1 +\mathfrak{v}_1 -\tfrac{1}{2} \Delta_2 + \mathfrak{v}_2 - \varepsilon_i - \varepsilon_j) \tilde{\sigma} \tau_{ij} \bigr)(0) \\
 & = & - M \bigl( \tilde{\sigma}' \sigma r^2 \mathfrak{Q} \frac{1}{|{\bf x}_1 - {\bf x}_2|} \Psi^{(1)}_{ij} \bigr) (0) \\
 & = & - M \bigl( \tilde{\sigma}' \sigma \mathfrak{Q} \frac{1}{|{\bf x}_1 - {\bf x}_2|} \Psi^{(1)}_{ij} \bigr) (2) ,
\end{eqnarray*}
which is due to the fact that the Mellin operator can be expressed as a local differential operator,
with cut-off functions $\tilde{\sigma}' \sigma \prec \tilde{\sigma}$, such that 
\[
 \bigl( -\tfrac{1}{2} \Delta_1 +\mathfrak{v}_1 -\tfrac{1}{2} \Delta_2 + \mathfrak{v}_2 - \varepsilon_i - \varepsilon_j \bigr) \tau_{ij} 
 = \mathfrak{Q} \frac{1}{|{\bf x}_1 - {\bf x}_2|} \Psi^{(1)}_{ij} ,
\]
is satisfied on the support of $\tilde{\sigma}' \sigma$.
The remaining term becomes
\begin{eqnarray*}  
 \int e^{iy \eta} M \bigl( \opm_M^{\gamma-1}(h_1^{(0)}) \tilde{\sigma} \hat{\tau}_{ij} \bigr) (0) \, \dbar \eta & = &
 M \bigl( \opm_M^{\gamma-1}(\tfrac{1}{2t^2}(w^2-w)) \tilde{\sigma} \tau_{ij} \bigr) (0) \\
 & = & \left. \tfrac{1}{2t^2}(w^2-w) M (\tilde{\sigma} \tau_{ij}) \right|_{w=0} \\
 & = & \left. \tfrac{1}{2t^2} \int_0^\infty r^{w-1} \biggl[ (-r \partial_r )^2 (\tilde{\sigma} \tau_{ij}) - (-r \partial_r ) (\tilde{\sigma} \tau_{ij}) \biggr] \right|_{w=0} dr \\
 & = & \tfrac{1}{2t^2} \int_0^\infty \biggl[ \partial_r \bigl( r \partial_r (\tilde{\sigma} \tau_{ij}) \bigr) + \partial_r (\tilde{\sigma} \tau_{ij}) \biggr] dr \\
 & = & - \tfrac{1}{2t^2} \left. \tau_{ij} \right|_{r=0} .
\end{eqnarray*}
It remains to calculate the action of the parametrix on the right hand side, i.e.,
\begin{eqnarray*}
 - P^{(ij)} \mathfrak{Q} \frac{1}{|{\bf x}_1 - {\bf x}_2|} \Psi^{(1)}_{ij} 
 & \sim & - \sigma \Op_y(p_M) \sigma' \mathfrak{Q} \frac{1}{|{\bf x}_1 - {\bf x}_2|} \Psi^{(1)}_{ij} \\ 
 & \sim & - \sigma r^2 \opm_M^{\gamma-3}(-2t^2h_0^{-1}) \sigma' \mathfrak{Q} \frac{1}{|{\bf x}_1 - {\bf x}_2|} \Psi^{(1)}_{ij} \\
 & \sim & \sigma r^2 2t^2 M^{-1}_\Gamma \frac{1}{(w-2)(w-3)} M \left( \sigma' \mathfrak{Q} \frac{1}{|{\bf x}_1 - {\bf x}_2|} \Psi^{(1)}_{ij} \right)(w) .
\end{eqnarray*}
Since $\tfrac{1}{2} < \gamma < \tfrac{3}{2}$, the integration has to performed along a line $\Gamma$, parallel to the complex axis, with $2 < \Re \Gamma < 3$.
The right hand side satisfies
\[
 \mathfrak{Q} \frac{1}{|{\bf x}_1 - {\bf x}_2|} \Psi^{(1)}_{ij} = \tfrac{1}{\sqrt{2}rt} \left. \Psi^{(1)}_{ij} \right|_{r=0} + {\cal O} (r^0) ,
\]
which means that the poles of its Mellin transform 
\begin{equation}
 M \left( \tilde{\sigma}' \mathfrak{Q} \frac{1}{|{\bf x}_1 - {\bf x}_2|} \Psi^{(1)}_{ij} \right)(w)
\label{MPsi}
\end{equation}
are located at integer values $w_0 \leq 1$.  
Therefore it is convenient to choose integration contours, depicted in Fig.~\ref{fig1}, such that the whole expression splits into
\begin{eqnarray}
\label{contour}
 - P^{(ij)} \mathfrak{Q} \frac{1}{|{\bf x}_1 - {\bf x}_2|} \Psi^{(1)}_{ij} 
 & \sim & \sigma' r^2 2t^2 \oint_{\Gamma_1} \frac{r^{-w}}{(w-2)(w-3)} M \left( \tilde{\sigma}' \mathfrak{Q} \frac{1}{|{\bf x}_1 - {\bf x}_2|} \Psi^{(1)}_{ij} \right)(w) \; \dbar w \\ \nonumber
 & & + \sigma' r^2 2t^2 \oint_{\Gamma_2} \frac{r^{-w}}{(w-2)(w-3)} M \left( \tilde{\sigma}' \mathfrak{Q} \frac{1}{|{\bf x}_1 - {\bf x}_2|} \Psi^{(1)}_{ij} \right)(w) \; \dbar w \\ \nonumber
 & \sim & \tfrac{1}{2} \sigma' (\sqrt{2}rt) \left. \Psi^{(1)}_{ij} \right|_{r=0} - 2t^2 \sigma'  M \left( \tilde{\sigma}' \mathfrak{Q} \frac{1}{|{\bf x}_1 - {\bf x}_2|} \Psi^{(1)}_{ij} \right)(2) .
\end{eqnarray}
The last term cancels with the corresponding term from the Green operator, yielding the asymptotic expansion 
\[
 {\cal P}_0 \tau_{ij} \sim \left. \tau_{ij} \right|_{r=0} + \tfrac{1}{2} (\sqrt{2}rt) \left. \Psi^{(1)}_{ij} \right|_{r=0} + {\cal O}(r^2) .
\]
In the same manner it is possible to calculate higher order contributions. With increasing order in $r$, the poles of the parametrix are shifted to the right. Again no overlap appears
between the poles of the parametrix and the poles of (\ref{MPsi}),
cf.~\cite{FFHS16} for further details. Therefore no logarithmic terms enter into the asymptotic expansion of the pair-amplitude.

\begin{figure}[t]
\begin{center}
\includegraphics[scale=0.5]{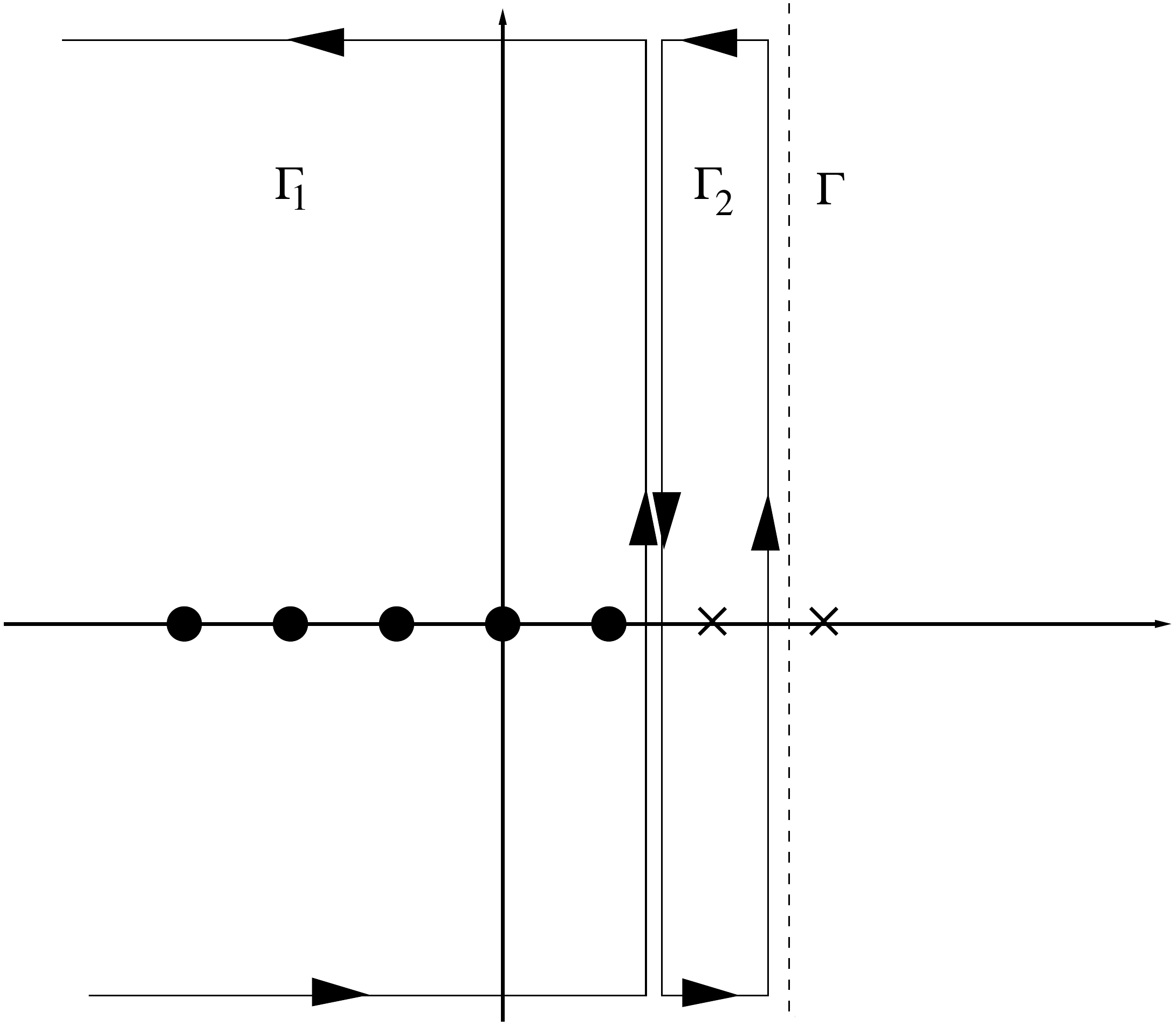}
\end{center}
\vspace{0.5cm}
\caption{Contours of integration used in Eq.~(\ref{contour}). Dots and crosses indicate poles of the meromorphic function (\ref{MPsi}) and (\ref{h0}), respectively.}
\label{fig1}
\end{figure}

Similar calculations can be performed for higher angular momentum values $l$ as well. Generally, poles of the parametrix are at $3+m+l$ and $2+m-l$, with $m=0,1,2,\ldots$, The second one moves to the left, however there is no coalescence with poles of (\ref{MPsi})
because of
\[
 {\cal P}_l \mathfrak{Q} \frac{1}{|{\bf x}_1 - {\bf x}_2|} \Psi^{(1)}_{ij} = {\cal O} (r^{l-1}) ,
\]
which means that the poles are located at $1-l-n$, with $n=0,1,2,\ldots$ Inclusion of the ladder terms (\ref{ladder}) and (\ref{RPA1}) on the right
hand side does not substantially alter the previous discussion, cf.~Propositions \ref{kernelasymp}, \ref{Wremainder} and Lemma \ref{lemmacompsymb},
leading to similar conclusions. Even so if we add the remaining
RPA terms to the right hand side these conclusions remain true. 
Let us summarize our considerations in the following
lemmas concerning iterated pair-amplitudes.

\begin{lemma}
\label{taun}
The iterated pair-amplitudes $\tau_n$, $n=0,1,2,\ldots$ of the RPA-CC equation or any other related model,
represented in hyperspherical coordinates $\tilde{\tau}_n$ belong to ${\cal S} \bigl( \mathbb{R}_+,H^{\infty,\gamma_1}_P(B) \bigr)$
and are of the generic asymptotic type without logarithmic terms. 
\end{lemma}

\begin{lemma}
Acting with the parameterix $P^{(ij)}$ on the RPA interactions
(\ref{RPA2}) and (\ref{RPA3}) corresponds to a mapping of the
symbols $\sigma^{\rpa}_{ij,-6}$ and $\sigma^{\rpa}_{ij,-10}$, cf.~(\ref{sigmaRPAij}),
from the symbol classes $S_{cl}^{-6}$ and $S_{cl}^{-10}$ into  
$S_{cl}^{-8}$ and $S_{cl}^{-12}$, respectively. This mapping
preserves the orthogonality constraint (\ref{orthtau4}) with $p=-8,-12$.
\label{propRPA}
\end{lemma}

\subsection{Proof of the main theorem}
\label{proofmaintheorem}
The proof of Theorem \ref{theorem1} and its Corollaries \ref{corollaryasymp} and \ref{corollaryfiltration} is a simple consequence of Lemmas
\ref{lemmatau}, \ref{lemmaV}, \ref{lemmacompsymb}, \ref{taun} and \ref{propRPA}. The statements of the theorem follow from
these lemmas by resolving the iterated pair-amplitudes into Goldstone diagrams and taking into account well known properties of the
calculus of classical pseudo-differential operators.

\section{Besov regularity of RPA diagrams}
\label{Besovregularity}
In Section \ref{RPAdiag} we have shown that RPA diagrams
can be considered within the algebra of classical pseudo-differential operators. The correspondence between classical symbols and kernel functions, cf.~Proposition \ref{SteinPr},
enables us to study the asymptotic behaviour of RPA diagrams
near coalescence points of electrons and to study adaptive
approximation schemes like best $N$-term approximation in hierarchical wavelet bases.
Previous results, presented in Ref.~\cite{FHS07},  on the best N-term approximation of two-particle correlation functions of Jastrow factors can be literally 
transfered to pair-amplitudes. What remains is a corresponding discussion for general RPA diagrams related to their symbol class. 
This is of potential interest with respect to the numerical simulation of RPA models. Let us just mention Corollary \ref{corollaryfiltration}, where it has been shown how
the symbol classes of RPA diagrams in iterative remainders vary with respect to the number of iteration steps. 

The concept of best $N$-term approximation belongs to the realm of nonlinear approximation theory.
For a detailed exposition of this subject we refer to Ref. \cite{DeVore}.
Loosely speaking, we consider for a given basis $\{ \zeta_i: i \in \Lambda \}$ the best possible
approximation of a function $f$ in the nonlinear subset $\Sigma_N$
which consists of all possible linear combinations of at most $N$ basis functions, i.e.,
\begin{equation}
 \Sigma_N := \left\{ \sum_{i \in \Delta} c_i \, \zeta_i : \Delta \subset \Lambda, \#\Delta \leq N, \,
 c_i \in \mathbb{R} \right\} .
\label{SigmaN}
\end{equation}
Here, the approximation error
\begin{equation}
 \sigma_N (f) := \inf_{f_N \in \Sigma_N} \| f - f_N \|_H 
\label{sigmaN}
\end{equation}
is given with respect to the norm of an appropriate Hilbert space $H$.
Best $N$-term approximation spaces $A^\alpha_q(H)$ for a Hilbert space $H$ can be defined according to
\begin{equation}
 A^\alpha_q(H) := \{f \in H : |f|_{A^\alpha_q(H)} < \infty \} \ \ \mbox{with} \ 
  |f|_{A^\alpha_q(H)} := \biggl( \sum_{N \in \mathbb{N}}
 \biggl(N^\alpha \sigma_N (f) \biggr)^q N^{-1} \biggr)^{\frac{1}{q}} .
\label{Aaq}
\end{equation}
It follows from the definition that a convergence rate $\sigma_N (f) \sim N^{-\alpha}$ with respect to
the number of basis functions can be achieved.

In our application, we consider anisotropic tensor product wavelets of the form
\begin{equation}
 \chi^{(\mathbf{s}_1, \mathbf{s}_2)}_{j_1,j_2,{\bf a}_1, {\bf a}_2} ({\bf x}_1,{\bf x}_2)
 = \gamma^{(s_1)}_{j_1,{\bf a}_1}({\bf x}_1) \, \gamma^{(s_2)}_{j_2,{\bf a}_2}({\bf x}_2) .
\label{chi2}
\end{equation}
These so called hyperbolic wavelets \cite{DKT} do not loose their efficiency in higher dimensions.
Each multivariate wavelet corresponds to an isotropic tensor product of orthogonal univariate wavelets
$\psi^{(1)}_{j,a}(x) := 2^{j/2} \psi(2^jx-a)$ and
scaling functions $\psi^{(0)}_{j,a}(x) := 2^{j/2} \varphi(2^jx-a)$ on the same level of refinement $j$, i.e.,
\begin{equation}
 \gamma^{({\bf s})}_{j,{\bf a}}({\bf x}) = \psi^{(s_1)}_{j,a_1}(x_1) \, \psi^{(s_2)}_{j,a_2}(x_2) \psi^{(s_3)}_{j,a_3}(x_3) 
 \ \ \mbox{with} \ {\bf s} :=(s_1,s_2,s_3) , \ {\bf a} :=(a_1,a_2,a_3) ,
\label{gamma}
\end{equation}
Pure scaling function tensor products $\gamma^{({\bf 0})}_{j_0,{\bf a}}$ are included on the coarsest level $j_0$ only.
For further details concerning wavelets, we refer to the monographs \cite{Daub,Mallat}.

Following Nitsche \cite{Nitsche2}, we consider tensor product Besov spaces
\[
 \tilde{{\cal B}}^\alpha_q(\Omega \times \Omega) =
 \left\{ B^{\alpha+1}_q(L_q(\Omega)) \bigotimes B^\alpha_q(L_q(\Omega)) \right\}
 \bigcap \left\{ B^\alpha_q(L_q(\Omega)) \bigotimes B^{\alpha+1}_q(L_q(\Omega)) \right\}
\]
for bounded domains $\Omega \subset \mathbb{R}^3$.
These spaces are norm equivalent to weighted $\ell_q$ norms for anisotropic wavelet coefficients
\begin{equation}
 \| f \|^q_{\tilde{{\cal B}}^\alpha_q} = \sum_{j_1,j_2 \geq j_0} 2^{\max\{j_1, j_2\} q}
 \left( \sum_{{\bf s}_1, {\bf s}_2} \sum_{{\bf a}_1, {\bf a}_2} \left|
 \langle \chi^{(\mathbf{s}_1, \mathbf{s}_2)}_{j_1,j_2,{\bf a}_1, {\bf a}_2} , f \rangle \right|^q \right) ,  
 \mbox{\ \ if \ } \alpha = \tfrac{3}{q} - \tfrac{3}{2} 
\label{tildeB}
\end{equation}
\[
 \mbox{with} \ \langle \chi^{(\mathbf{s}_1, \mathbf{s}_2)}_{j_1,j_2,{\bf a}_1, {\bf a}_2} , f \rangle :=
 \int_{\mathbb{R}^3 \times \mathbb{R}^3} \gamma^{({\bf s}_1)}_{j_1,{\bf a}_1} (\mathbf{x}) \,
 f(\mathbf{x}, \mathbf{y}) \, \gamma^{({\bf s}_2)}_{j_2,{\bf a}_2} (\mathbf{y}) \, d\mathbf{x} d\mathbf{y} . 
\]
The norm equivalence requires a univariate wavelet $\psi$ with $p > \alpha +1$ vanishing moments and
$\psi, \varphi \in B^\beta_q(L_q)$ for some $\beta > \alpha+1$.
The corresponding relation between best $N$-term approximation spaces and Besov spaces is given by
\begin{equation}
 A^{\alpha/3}_q(H^1(\Omega \times \Omega)) = \tilde{{\cal B}}^\alpha_q(\Omega \times \Omega) , \ \ 
 \mbox{if} \ \alpha = \tfrac{3}{q} - \tfrac{3}{2} .
\label{ABdH2}
\end{equation}
The following lemma  provides the Besov regularity of RPA diagrams depending on their symbol class, and according
to (\ref{ABdH2}), anticipated convergence rates of adaptive approximation schemes. It should be mentioned that
the following bounds concerning Besov regularities are sharp and cannot be improved. This follows from a simple
argument, cf.~Ref.~\cite[Corollary 2.4]{FHS07}, which can be easily adapted to the present case.

\begin{lemma}
\label{besov1}
Let $\tau_{\rpa}$ represent a RPA Goldstone diagram with corresponding symbol in $S^p_{cl}(\mathbb{R}^3 \times \mathbb{R}^3)$, $p \leq -4$. Then 
$\tau_{\rpa}$
belongs to $\tilde{{\cal B}}^\alpha_q(\Omega \times \Omega)$ for $q > -\frac{3}{1+p}$ and $\alpha = \frac{3}{q} - \frac{3}{2}$. 
\end{lemma}

\begin{proof}
The following proof is closely related to the proof of Lemma 2.1 in Ref.~\cite{FHS07}, 
For each isotropic 3d-wavelet $\gamma^{({\bf s})}_{j,{\bf a}}$, we define a cube $\Box_{j,{\bf a}}$
centred at $2^{-j} {\bf a}$ with edge length $2^{-j} L$, such that $\supp \gamma^{({\bf s})}_{j,{\bf a}}
\subset \Box_{j,{\bf a}}$. In order to estimate the norm (\ref{tildeB}) for a RPA diagram,
we restrict ourselves to wavelet coefficients with $j_1 \geq j_2$ and $|\mathbf{s}_1|=|\mathbf{s}_2| = 1$.
This combination of 3d-wavelet types corresponds to the worst case where vanishing moments can act
in one direction only.  

We first consider the case $\dist(\Box_{j_1,{\bf a}_1}, \Box_{j_2,{\bf a}_2}) \leq 2^{-j_2} L$.
In order to apply the asymptotic smoothness property, c.f.~Proposition \ref{SteinPr},
\begin{align}
\label{asympsmooth1}
\left| \partial_{\bf x}^{\beta} \partial_{\bf z}^{\alpha}
\tau_{\rpa} \right| \lesssim 1 &
\quad \mbox{for} \ |\alpha| \leq -3-p , \\ \label{asympsmooth2}
\left| \partial_{\bf x}^{\beta} \partial_{\bf z}^{\alpha}
\tau_{\rpa} \right| \lesssim |{\bf z}|^{-3-p-|\alpha|} &
\quad \mbox{for} \ |\alpha| > -3-p ,
\end{align}
let us decompose the cube $\Box_{j_2,{\bf a}_2}$
into non overlapping subcubes $\Box_i$ $(i \in \Delta)$ with edge length $2^{-j_1} L$. 
The subcubes $\Box_i$ with $i \in \Delta_0 := \{i \in \Delta : \dist(\Box_{j_1,{\bf a}_1}, \Box_i) \leq
2^{-j_1} L \}$ are considered separately. Their number is $\#\Delta_0 = O(1)$ independent of the wavelet levels $j_1, j_2$.
For the remaining subcubes $\Box_i$ $(i \in \Delta \setminus \Delta_0)$
it becomes necessary to control their contributions with respect to $\dist(\Box_{j_1,{\bf a}_1}, \Box_i)$
because $\#(\Delta \setminus \Delta_0) = O(2^{3(j_1-j_2)})$ depends on the wavelet levels.
The wavelet coefficients can be estimated by the separate sums
\begin{eqnarray}
\nonumber
 \left| \langle \chi^{(\mathbf{s}_1, \mathbf{s}_2)}_{j_1,j_2,{\bf a}_1, {\bf a}_2} , \tau_{\rpa} \rangle \right|
 & \leq & \sum_{i \in \Delta_0}
 \left| \int_{\Box_{j_1,{\bf a}_1} \times \Box_i} \gamma^{({\bf s}_1)}_{j_1,{\bf a}_1} (\mathbf{x}) \,
 \tau_{\rpa} (\mathbf{x},\mathbf{x} - \mathbf{y}) \, \gamma^{({\bf s}_2)}_{j_2,{\bf a}_2} (\mathbf{y}) 
 d\mathbf{x} d\mathbf{y} \right| \\ \label{deltaF}
 & + & \sum_{i \in \Delta \setminus \Delta_0} 
 \left| \int_{\Box_{j_1,{\bf a}_1} \times \Box_i} \gamma^{({\bf s}_1)}_{j_1,{\bf a}_1} (\mathbf{x}) \,
 \tau_{\rpa} (\mathbf{x},\mathbf{x} - \mathbf{y}) \, \gamma^{({\bf s}_2)}_{j_2,{\bf a}_2} (\mathbf{y})
 d\mathbf{x} d\mathbf{y} \right| .
\end{eqnarray}
For the first sum we can use the following proposition
\begin{proposition}
\label{HM}
The RPA diagram $\tau_{\rpa}$ satisfies the estimate
\[
\left| \int_{\mathbb{R}^{3}} \tau_{\rpa}({\bf x},{\bf x}-{\bf y}) \gamma^{({\bf s})}_{j,{\bf a}}({\bf y}) d{\bf y} \right|
\lesssim 2^{-j(m-3)} \| \tau_{\rpa}({\bf x}, \cdot )
\|_{H^{m}(\mathbb{R}^{3})} ,
\]
for $\dist({\bf x}, \Box_{j,{\bf a}}) \leq 2^{-j} L$ and any integer $m$ such that $m < \min \{-\frac{1}{2} -p,q\}$.
\end{proposition}
\begin{proof}
It is an immediate consequence of the symbol estimate, see e.g.~\cite{Stein},
\[
 \left| \partial_{\bf x}^{\beta} \partial_{\boldeta}^{\alpha}
\sigma_{\rpa} ({\bf x},\boldeta)\right| \lesssim \bigl( 1+|\boldeta|\bigr)^{p-|\alpha|} ,
\]
that $\tau_{\rpa}({\bf x}, \cdot )$ belongs to the
Sobolev space $H^{s}(\mathbb{R}^{3})$ for $s < -\frac{1}{2} -p$.
Similar to the proof of Proposition \ref{kpN}, let us decompose
the pair-amplitude into a polynomial and a singular remainder
\[
 \tau_{\rpa}({\bf x},{\bf z}) = Q_{p}({\bf x},{\bf z}) + R_{p}^m({\bf x},{\bf z}) ,
\]
with $Q_{p}$ a polynomial of degree less than $m$ in ${\bf z} := {\bf x} - {\bf y}$ with coefficients which are smooth functions in ${\bf x}$.
For the remainder we can achieve the following estimate, cf.~Prop.~4.3.2 \cite{BS},
\[
 \| R_{p}^m({\bf x},{\bf x} - \cdot) \|_{L^\infty(\Box_{j,{\bf a}})} \lesssim 2^{-(m-\frac{3}{2})j} \| R_{p}^m({\bf x},{\bf x} - \cdot) \|_{H^m(B)} ,
\]
where $B$ is a ball centered at ${\bf x}$ with radius $2^{-j+3}L$.
Taking into account the normalization of the wavelet, we obtain the desired estimate.
\end{proof}
Suppose $s_{1,1} =1$ for the wavelet 
$\gamma^{({\bf s}_1)}_{j_1,{\bf a}_1}$, cf.~(\ref{gamma}), we obtain from Proposition \ref{HM} the estimate
\begin{eqnarray}
\nonumber
 \lefteqn{\sum_{i \in \Delta_0}
 \left| \int_{\Box_{j_1,{\bf a}_1} \times \Box_i} \gamma^{({\bf s}_1)}_{j_1,{\bf a}_1} (\mathbf{y}) \,
 \tau_{\rpa} (\mathbf{x}, \mathbf{x} - \mathbf{y}) \, \gamma^{({\bf s}_2)}_{j_2,{\bf a}_2} (\mathbf{x})
 d\mathbf{x} d\mathbf{y} \right|} \hspace{5cm} \\ \nonumber
 & \lesssim & 2^{-(m+3) j_1} 2^{\frac{3}{2} j_2}
 \| \tau_{\rpa}({\bf x}, \cdot ) \|_{H^{m}(\mathbb{R}^{3})} \\ \label{est1}
  & \lesssim & \left\{ \begin{array}{cl} 2^{-(q+3) j_1} 2^{\frac{3}{2} j_2}, & \mbox{\ if \ } q < -\frac{1}{2} -p\\
2^{(p-2) j_1} 2^{\frac{3}{2} j_2}, & \mbox{\ if \ } q > -\frac{1}{2} -p \end{array} \right. .
\end{eqnarray}

The second sum can be estimated using the next proposition (see, e.g., Ref. \cite{FHS07} for details).
\begin{proposition}
\label{estimates}
Suppose the function $f(\mathbf{x})$ with $\mathbf{x} \in \mathbb{R}^3$ is smooth on the support of an
isotropic 3d-wavelet $\gamma^{(\mathbf{s})}_{j,\mathbf{a}}$. Then the following estimate holds
\[
    \left| \int_{\mathbb{R}^3} f ( \mathbf{x} ) \, \gamma^{(\mathbf{s})}_{j,\mathbf{a}} ( \mathbf{x}
    ) d\mathbf{x} \right| \lesssim  2^{- (p |\mathbf{s}|+3/2) j}
    \| \partial_{x_1}^{s_1 p} \partial_{x_2}^{s_2 p} \partial_{x_3}^{s_3 p} f \|_{L_{\infty} (
 \supp \gamma^{(\mathbf{s})}_{j,\mathbf{a}})} 
    \ \ \mbox{with} \ |\mathbf{s}| := s_1 + s_2 + s_3 .
\]
\label{estdd}
\end{proposition}
With this and the estimates (\ref{asympsmooth1}) and (\ref{asympsmooth2}) for wavelets $\gamma^{({\bf s}_1)}_{j_1,{\bf a}_1}$ 
with $q$ vanishing moments (i.e., $|{\bf s}_1| =1$), we obtain
\begin{eqnarray}
\nonumber
 \lefteqn{\sum_{i \in \Delta \setminus \Delta_0}
 \left| \int_{\Box_{j_1,{\bf a}_1} \times \Box_i} \gamma^{({\bf s}_1)}_{j_1,{\bf a}_1} (\mathbf{y}) \,
 \tau_{\rpa} (\mathbf{x}, \mathbf{x} - \mathbf{y}) \, \gamma^{({\bf s}_2)}_{j_2,{\bf a}_2} (\mathbf{x}) d\mathbf{x} d\mathbf{y} \right|} \\ \nonumber
 & & \lesssim \sum_{i \in \Delta \setminus \Delta_0} 
 2^{-(q+\frac{3}{2}) j_1} 2^{-3 j_1} 2^{\frac{3}{2} j_2}
 \| \partial^q_{y_1} \tau_{\rpa} \|_{L_\infty(\Box_{j_1,{\bf a}_1} \times \Box_i)} \\ \label{est2}
 & & \lesssim 2^{-(q+ \frac{3}{2}) j_1} 2^{\frac{3}{2} j_2} \int_{2^{-j_1} L}^{2^{-j_2 +2} L}
 r^{2-k} \, dr
 \lesssim \left\{ \begin{array}{cl} 2^{-(q+ \frac{3}{2}) j_1} 2^{-(\frac{3}{2}-k) j_2} , & \mbox{\ if \ } k \leq 2 \\
                 (j_1 - j_2 +1) 2^{-(\frac{3}{2} -p) j_1} 2^{\frac{3}{2} j_2} , & \mbox{\ if \ } k = 3 \\
                                2^{-(\frac{3}{2} -p) j_1} 2^{\frac{3}{2} j_2} , & \mbox{\ if \ } k > 3 \\
              \end{array} \right. , 
\end{eqnarray}
where the parameter $k=0$ if $q \leq -3-p$ and $k=3+p+q$ otherwise.

Once we have obtained the estimates (\ref{est1}) and (\ref{est2}), it is straightforward to get an upper bound 
for the contribution of anisotropic tensor products with translation parameters
$({\bf a}_1,{\bf a}_2) \in A_{j_1,j_2} := \{ ({\bf a}_1,{\bf a}_2) : 
\dist(\Box_{j_1,{\bf a}_1}, \Box_{j_2,{\bf a}_2}) \leq 2^{-j_2} L\}$ to the norm (\ref{tildeB})
\begin{eqnarray}
\nonumber
 \sum_{j_1 \geq j_2 \geq j_0} 2^{\tilde{q}j_1} \sum_{({\bf a}_1, {\bf a}_2) \in A_{j_1,j_2}}
 \left| \langle \chi^{(\mathbf{s}_1, \mathbf{s}_2)}_{j_1,j_2,{\bf a}_1, {\bf a}_2} , \tau_{\rpa} \rangle \right|^{\tilde{q}}
 & \lesssim & \sum_{j_1 \geq j_2 \geq j_0}
 \left\{ \begin{array}{cl} 2^{-(q\tilde{q} + \frac{\tilde{q}}{2} -3) j_1} 2^{-(\frac{3}{2} -k)q j_2} , & \mbox{\ if \ } k \leq 2 \\
                      (j_1 - j_2 +1) 2^{-(\frac{1}{2} \tilde{q} -p\tilde{q} -3) j_1} 2^{\frac{3}{2} \tilde{q} j_2} , & \mbox{\ if \ } k = 3 \\
                                     2^{-(\frac{1}{2} \tilde{q} -p\tilde{q} -3) j_1} 2^{\frac{3}{2} \tilde{q} j_2} , & \mbox{\ if \ } k > 3 \\
            \end{array} \right. \\ \label{estA}
 & < & \infty , \mbox{\ if \ } \left\{ \begin{array}{cl} \tilde{q} > \frac{3}{q+\frac{1}{2}} , & \mbox{\ for \ } k \leq 1 \\
                                      \tilde{q} > -\frac{3}{1+p} , & \mbox{\ for \ } k \geq 2 \end{array} \right. ,
\end{eqnarray}
where we have used $\# A_{j_1,j_2} = O(2^{3j_1})$. 

In order to get an upper bound for the norm (\ref{tildeB}) it remains to estimate the contributions of anisotropic
wavelet coefficients where the supports of the corresponding 3d-wavelets are well separated. For this we have to
consider the parameter set
$B_{j_1,j_2} := \{ ({\bf a}_1,{\bf a}_2) : 2^{-j_2} L < \dist(\Box_{j_1,{\bf a}_1}, \Box_{j_2,{\bf a}_2})\}$.
Let us assume $q > -p-3$, i.e.~$k=3+p+q$, using estimates (\ref{est1}) and (\ref{est2}) and Proposition \ref{estimates}, the contributions can be estimated by
\begin{eqnarray}
\nonumber
 \sum_{({\bf a}_1, {\bf a}_2) \in B_{j_1,j_2}}
 \left| \langle \chi^{(\mathbf{s}_1, \mathbf{s}_2)}_{j_1,j_2,{\bf a}_1, {\bf a}_2} , \tau_{\rpa} \rangle \right|^{\tilde{q}}
 & \lesssim & \sum_{({\bf a}_1, {\bf a}_2) \in B_{j_1,j_2}}
 2^{-(q+ \frac{3}{2})\tilde{q}(j_1 +j_2)} \| \partial^p_{x_1} \partial^p_{y_1} \tau_{\rpa}
 \|_{L_\infty(\Box_{j_1,{\bf a}_1} \times \Box_{j_2,{\bf a}_2})}^{\tilde{q}} \\ \nonumber
 & \lesssim & 2^{-(q\tilde{q} + \frac{3}{2}\tilde{q} -3)j_1} 2^{-(q\tilde{q}+ \frac{3}{2}\tilde{q}-3)j_2}
 \int_{2^{-j_2} L}^{\diam \Omega} r^{2-\tilde{q}(3+p+2q)} \, dr \\ \nonumber
 & \lesssim & 2^{-(q\tilde{q} + \frac{3}{2}\tilde{q} -3) j_1} 2^{(q+p+ \frac{3}{2})\tilde{q} j_2} ,
\end{eqnarray}
where we have used $p > \alpha +1$ and $\alpha =\frac{3}{q} - \frac{3}{2}$, hence 
\[
 2 -2q\tilde{q} -3\tilde{q} -p\tilde{q} < -1 -\tilde{q} \biggl( q+p+ \tfrac{5}{2} \biggr) < -1 \ \mbox{since} \  p+q >-3, \ p,q \in \mathbb{Z} ,
\]
follows for the exponent of the integrand. 
The remaining sum with respect to the wavelet levels yields
\begin{eqnarray*}
 \lefteqn{\sum_{j_1 \geq j_2 \geq j_0} 2^{\tilde{q}j_1} \sum_{({\bf a}_1, {\bf a}_2) \in B_{j_1,j_2}}
 \left| \langle \chi^{(\mathbf{s}_1, \mathbf{s}_2)}_{j_1,j_2,{\bf a}_1, {\bf a}_2} , \tau_{\rpa} \rangle \right|^{\tilde{q}}
  \lesssim \sum_{j_1 \geq j_2 \geq j_0} 2^{-(q\tilde{q}+ \frac{1}{2}\tilde{q} -3) j_1} 2^{(q+p+ \frac{3}{2})\tilde{q} j_2}} \\ 
 & & \lesssim \sum_{j_1 \geq j_0} \left\{ \begin{array}{cl} 2^{-(q\tilde{q}+ \frac{1}{2}\tilde{q} -3) j_1} , & \mbox{\ if \ } p+q = -2 \\
                                           2^{(3+\tilde{q}p+\tilde{q}) j_1} , & \mbox{\ if \ } p+q \geq -1 \\ \end{array} \right.
 < \infty , \mbox{\ if \ } \left\{ \begin{array}{cl} \tilde{q} > \frac{3}{q+\frac{1}{2}} , & \mbox{\ for \ } p+q = -2 \\
                                      \tilde{q} > -\frac{3}{1+p} , & \mbox{\ for \ } p+q \geq -1 \\ \end{array} \right. ,
\end{eqnarray*}
from which we obtain, together with our previous estimate (\ref{estA}), the lower bound on the Besov space parameter $q$.
\end{proof}

\section{Acknowledgement}
Financial support from the Deutsche Forschungsgemeinschaft DFG (Grant No.~HA 5739/3-1) is gratefully acknowledged.

\appendix
\vspace{1cm}
\noindent {\Large {\bf Appendix}}

\section{Weighted edge Sobolev spaces with asymptotics} 
\label{appendix1}
In order to incorporate asymptotics into Sobolev spaces one has to proceed in a recursive manner.
Let us first consider
weighted Sobolev spaces ${\cal K}^{s,\gamma}(X^\wedge)$ on an open stretched cone with base $X$, which are defined with respect to the corresponding polar coordinates $\tilde{x} \rightarrow (r,x)$ via
\[
 {\cal K}^{s,\gamma}(X^\wedge) := \omega {\cal H}^{s,\gamma}(X^\wedge) +(1-\omega) H^s(\mathbb{R}^{3}) ,
\]
for a cut-off function $\omega$, i.e., $\omega\in C_0^\infty(\overline{\mathbb{R}}_+)$ such that $\omega(r)=1$ near $r=0$.
Here ${\cal H}^{s,\gamma}(X^\wedge) = r^\gamma {\cal H}^{s,0}(X^\wedge)$, and ${\cal H}^{s,0}(X^\wedge)$
for $s \in \mathbb{N}_0$ is defined to be the set of all $u(r,x) \in r^{-1} L^2(\mathbb{R}_+ \times X)$
such that $(r \partial_r)^jDu \in r^{-1} L^2(\mathbb{R}_+ \times X)$ for all $D \in \mbox{Diff}^{s-j}(X)$,
$0 \leq j \leq s$. The definition for $s \in \mathbb{R}$ in general follows by duality and complex interpolation.
Weighted Sobolev spaces with asymptotics are subspaces of ${\cal K}^{s,\gamma}$ spaces which are defined as direct sums
\begin{equation}
 {\cal K}^{s,\gamma}_Q (X^\wedge) := {\cal E}^\gamma_Q (X^\wedge) + {\cal K}^{s,\gamma}_\Theta (X^\wedge) 
\label{E+K}
\end{equation}
of flattened weighted cone Sobolev spaces
\[
 {\cal K}^{s,\gamma}_\Theta (X^\wedge) := \bigcap_{\epsilon > 0} {\cal K}^{s,\gamma - \vartheta - \epsilon}
 (X^\wedge) 
\]
with $\Theta =(\vartheta,0]$, $-\infty \leq \vartheta < 0$, and
asymptotic spaces
\[
 {\cal E}^\gamma_Q (X^\wedge) := \biggl\{ \omega(r) \sum_j \sum_{k=0}^{m_j} c_{jk}(x) r^{-q_j} \ln^k r \biggr\} .
\]
The asymptotic space ${\cal E}^\gamma_Q (X^\wedge)$ is characterized by a sequence $q_j \in \mathbb{C}$
which is taken from a strip of the complex plane, i.e.,
\[
 q_j \in \left\{ z: \frac{3}{2}-\gamma + \vartheta < \Re z < \frac{3}{2}-\gamma \right\} ,
\]
where the width and location of this strip are determined by its {\em weight data} $(\gamma,\Theta)$
with $\Theta =(\vartheta,0]$ and $-\infty \leq \vartheta < 0$. Each substrip of finite width
contains only a finite number of $q_j$. Furthermore, the coefficients
$c_{jk}$ belong to finite dimensional subspaces $L_j \subset C^\infty(X)$.
The asymptotics of ${\cal E}^\gamma_Q(X^\wedge)$ is therefore completely
characterized by the {\em asymptotic type} $Q := \{(q_j,m_j,L_j)\}_{j \in \mathbb{Z}_+}$.
In the following, we employ the asymptotic subspaces
\[
 {\cal S}^\gamma_Q (X^\wedge) := \left\{ u \in {\cal K}^{\infty,\gamma}_Q (X^\wedge) :
 (1- \omega) u \in {\cal S}(\mathbb{R},C^\infty(X))|_{\mathbb{R}_+} \right\} 
\]
with Schwartz type behaviour for exit $r \rightarrow
\infty$. The spaces ${\cal K}^{s,\gamma}_Q(X^\wedge)$ and ${\cal
S}^\gamma_Q (X^\wedge)$ are Fr\'echet spaces equipped with
natural semi-norms according to the decomposition (\ref{E+K}); we
refer to \cite{ES97,HS08,Schulze98} for further details.

Weighted wedge Sobolev spaces on $\mathbb{W} := X^{\wedge} \times Y$ can be defined as
functions $Y \rightarrow {\cal K}_{(Q)}^{s, \gamma}(X^\wedge)$, where a subscript $Q$ optionally denotes cone spaces
with asymptotics. Let us first consider the case $Y= \mathbb{R}^q$ and corresponding wedge Sobolev spaces
\[
 {\cal W}^s(\mathbb{R}^q, {\cal K}_{(Q)}^{s, \gamma}(X^\wedge)) :=
 \{ u : \mathbb{R}^q \rightarrow {\cal K}_{(Q)}^{s, \gamma}(X^\wedge)
 \, | \, u \in  \overline{{\cal S}(\mathbb{R}^q, {\cal K}_{(Q)}^{s, \gamma}(X^\wedge)} \} 
\]
with $s, \gamma \in \mathbb{R}$ and norm closure w.r.t.~the norm
\[
 \| u \|_{{\cal W}^{s}(\mathbb{R}^q, {\cal K}_{(Q)}^{s, \gamma}(X^\wedge))}^2 := \int [\eta]^{2s} \| \kappa^{-1}_{[\eta]} (F_{y\rightarrow \eta} u)(\eta)
 \|_{{\cal K}^{s,\gamma}_{(Q)}(X^\wedge)}^2 d\eta .
\]
Here $F_{y\rightarrow \eta}$ denotes the Fourier transform in $\mathbb{R}^q$ and $\{ \kappa_\lambda \}_{\lambda \in \mathbb{R}_+}$
a strongly continuous group of isomorphisms $\kappa_\lambda : {\cal K}_{(Q)}^{s, \gamma}(X^\wedge) \rightarrow 
{\cal K}_{(Q)}^{s, \gamma}(X^\wedge)$ defined by
\[
 \kappa_\lambda u(r,x,y) := \lambda^{\frac{3}{2}} u(\lambda r,x,y) .
\]
The function $[\eta]$ involved in the norm is given by a strictly positive $C^{\infty}(\mathbb{R}^{3})$ function of the covariables $\eta$ such that
$[\eta] = |\eta|$ for $|\eta| \geq \epsilon >0$. The motivation behind this group action is the twisted homogeneity
of principal edge symbols, cf.~\cite{Schulze98} for further details. For $Y \subset \mathbb{R}^q$ an open subset, we define
\[
 {\cal W}^s_{\mbox{\footnotesize comp}}(Y, {\cal K}_{(Q)}^{s, \gamma}(X^\wedge)) :=
 \{ u \in {\cal W}^s(\mathbb{R}^q, {\cal K}_{(Q)}^{s, \gamma}(X^\wedge)): \supp u \subset Y \ \mbox{compact} \} ,
\]
and
\[
 {\cal W}^s_{\mbox{\footnotesize loc}}(Y, {\cal K}_{(Q)}^{s, \gamma}(X^\wedge)) :=
 \{ u \in {\cal D}'(Y, {\cal K}_{(Q)}^{s, \gamma}(X^\wedge)): \varphi u \in
 {\cal W}^s_{\mbox{\footnotesize comp}}(\mathbb{R}^q, {\cal K}_{(Q)}^{s, \gamma}(X^\wedge)) \ \mbox{for each} \ 
 \varphi \in C^\infty_0(Y) \} .
\]
The weighted Sobolev spaces ${\cal W}^{\infty}_{\mbox{\footnotesize comp}}(Y, {\cal K}_{(Q)}^{\infty, \gamma}(X^\wedge))$, which are of particular interest 
in our application, have a nice tensor product representation for their asymptotic expansion given by 
\begin{equation}
 \omega(r) \sum_{j} \sum_{k=0}^{m_j} r^{-p_j} \log^k r \, c_{jk}(x) v_{jk}(y) + h_{\Theta}(r,x,y)
\label{appasymp}
\end{equation}
where $(r,x,y)$ denote appropriate coordinates on the wedge $X^\wedge \times Y$. Tensor components $c_{jk} \in C^\infty(X)$, $v_{jk} \in H^\infty_{\mbox{\footnotesize comp}}(Y)$
correspond to functions on the base of the cone $X$ and the edge $Y$, respectively. This tensor product decomposition represents an essential part of our approach 
and has been  frequently applied in the present work.

\end{document}